\documentclass[letterpaper]{article} 

\usepackage[letterpaper]{geometry}
\usepackage[parfill]{parskip}
\PassOptionsToPackage{numbers, compress}{natbib}
\usepackage[numbers, compress]{natbib}
\usepackage{xr-hyper,refcount}

\usepackage{graphicx}
\usepackage[utf8]{inputenc} 
\usepackage[T1]{fontenc}    
\usepackage{url}  
\usepackage{tabularx}
\usepackage{booktabs}       
\usepackage{amsfonts}       
\usepackage{nicefrac}       
\usepackage{microtype}      
\usepackage[parfill]{parskip}
\usepackage{algorithm,algorithmic}
\usepackage{amsmath,amsthm,amssymb,bbm}
\usepackage{mathtools}
\usepackage{cases}
\usepackage{comment}
\usepackage{subcaption}
\usepackage{algorithm,algorithmic}
\usepackage{color}
\usepackage{appendix}
\usepackage{xspace}
\usepackage{enumitem}
\usepackage[english]{babel}
\usepackage{authblk}

\usepackage[colorlinks,citecolor=blue,urlcolor=blue,linkcolor=blue,linktocpage=true]{hyperref}
\usepackage{cleveref}
\crefformat{equation}{(#2#1#3)}
\crefrangeformat{equation}{(#3#1#4) to~(#5#2#6)}
\crefname{equation}{}{}
\Crefname{equation}{}{}

\crefname{definition}{\textbf{definition}}{definitions}
\Crefname{definition}{Definition}{Definitions}
\crefname{assumption}{\textbf{assumption}}{assumptions}
\Crefname{assumption}{Assumption}{Assumptions}
\definecolor{maroon}{RGB}{192,80,77}

\newtheorem{theorem}{Theorem}
\newtheorem{lemma}{Lemma}

\newtheorem{proposition}{Proposition}

\newtheorem{definition}{Definition}

\def\sign{\mathrm{sign}}

\newenvironment{pfsketch}{%
  \proof}{\endproof}

	\AtEndDocument{\refstepcounter{theorem}\label{finalthm}}

\usepackage{mdwlist}

\def\sign{\mathrm{sign}}

\newcommand{\adv}{majority}

\newcommand{\dis}{minority}

\newcommand{\A}{A}
\newcommand{\D}{D}

\newcommand{\vecl}[1]{\mathbf{#1}}
\renewcommand{\vec}[1]{\boldsymbol{#1}}

\newcommand{\xhdr}[1]{\vspace{1mm} \noindent\textbf{#1.}}

\title{Moral Machine or Tyranny of the Majority?}

\begin{document}

\author{Michael Feffer}
\author{Hoda Heidari\footnote{These authors contributed equally.}}
\newcommand\CoAuthorMark{\footnotemark[\arabic{footnote}]} 
\author{Zachary C. Lipton\protect\CoAuthorMark}
\affil{
    Carnegie Mellon University\\
   \{\href{mailto:mfeffer@andrew.cmu.edu}{\nolinkurl{mfeffer}}, \href{mailto:hheidari@andrew.cmu.edu}{\nolinkurl{hheidari}},
   \href{mailto:zlipton@andrew.cmu.edu}{\nolinkurl{zlipton}}\}@andrew.cmu.edu
}


 
 
 
 

\date{}

\maketitle

\begin{abstract}
With Artificial Intelligence systems increasingly 
applied in consequential domains,
researchers have begun to ask how these systems 
ought to act in ethically charged situations
where even humans lack consensus. 
In the Moral Machine project, researchers crowdsourced answers 
to ``Trolley Problems'' concerning autonomous vehicles.
Subsequently, 
Noothigattu et al. (2018)
proposed 
inferring linear functions that approximate 
each individual's preferences
and aggregating these linear models 
by averaging parameters across the population.
In this paper, we examine this \emph{averaging} mechanism,
focusing on fairness concerns in the presence of strategic effects.
We investigate a simple setting
where the population consists of two groups,
with the minority constituting an $\alpha < 0.5$ 
share of the population.
To simplify the analysis,
we consider the extreme case in which
within-group preferences are homogeneous.
Focusing on the fraction of contested cases
where the minority group prevails,
we make the following observations:
(a) even when all parties report their preferences truthfully,
the fraction of disputes where the minority prevails
is less than proportionate in $\alpha$;
(b) the degree of sub-proportionality grows more severe
as the level of disagreement between the groups increases;
(c) when parties report preferences strategically,
pure strategy equilibria do not always exist;
and (d) whenever a pure strategy equilibrium exists,
the majority group prevails 100\% of the time. 
These findings raise concerns 
about stability and fairness 
of preference vector averaging as a mechanism 
for aggregating diverging voices.
Finally, we discuss alternatives,
including randomized dictatorship 
and median-based mechanisms. 

\end{abstract}

\section{Introduction}
\label{sec:intro}
Machine learning (ML) has increasingly been employed
to automate decisions in consequential domains 
including autonomous vehicles, healthcare, 
hiring, finance, and criminal justice.
These domains present many ethically charged decisions,
where even knowledgeable humans may lack consensus
about the right course of action.
Consequently, AI researchers have been forced 
to consider how to resolve normative disputes
when competing values come into conflict.
The formal study of such ethical quandaries 
long predates the advent of modern AI systems.
For example, philosophers have long debated
Trolley Problems \citep{thomson1985trolley},
which generally take the form of inescapable decisions 
among normatively undesirable alternatives.
Notably, these problems typically lack clear-cut answers,
and people's judgments are often sensitive 
to subtle details in the provided context. 

Questions about whose values are represented
and what objectives are optimized 
in ML-based systems
have become especially salient 
in light of documented instances 
of algorithmic bias in deployed systems
\citep{angwin2016machine, buolamwini2018gender, obermeyer2019dissecting}. 
Faced with questions of whose values ought to prevail
when a judgment must be made,
some AI ethics researchers have advocated 
\emph{participatory machine learning},
a family of methods for democratizing decisions
by incorporating the views of a variety of stakeholders
\citep{lee2019webuildai, ilvento2019metric, jung2019algorithmic}.
For example, a line of papers on \emph{preference elicitation}
tasks stakeholders with choosing among sets of alternatives
\citep{ilvento2019metric, lee2019webuildai, jung2019algorithmic, hiranandani2020fair, hiranandani2020quadratic, freedman2020adapting}.
In many of these studies, the hope is 
to compute a socially aligned objective function 
\citep{lee2019webuildai, freedman2020adapting}
or fairness metric \citep{ilvento2019metric, hiranandani2020quadratic, hiranandani2020fair}.

In a pioneering study,
\citet{awad2018moral} introduced the Moral Machine,
a large-scale crowdsourcing study 
in which millions of participants from around the world
were presented 
with autonomous driving scenarios
in the style of Trolley Problems. 
Participants were shown images 
depicting two possible outcomes 
and asked which alternative they preferred.
In one scenario, the first alternative 
might be to collide with a barrier 
and sacrifice several young passengers, 
and the second to pulverize 
two elderly pedestrians in a crosswalk. 
Utilizing the Moral Machine dataset,
\citet{noothigattu2018voting} 
proposed methods for inferring 
each participant's preferences.
Specifically, they represent each alternative 
by a fixed-length vector of attributes, $\vecl{x}$, 
and each participant's scoring function
as the dot product between their \emph{preference vector} $\vec{\theta}$
and the alternative $\vecl{x}$.
The objective is to infer parameters $\vec{\theta}$ 
so that whenever an individual 
prefers one alternative over another,
the preferred alternative 
receives a higher score.
To aggregate these preference vectors across a population, 
\citet{noothigattu2018voting} propose 
to simply \emph{average} them.
Faced with a ``moral dilemma'', an autonomous vehicle
would conceivably featurize each alternative,
compute their dot products 
with the \emph{aggregate preference vector},
and choose the highest scoring option.
These approaches have been echoed
in a growing body of follow-up studies
\citep{pugnetti2018customer, wang2019privacy,kemmer2020enhancing}.
However, despite the work's influence, 
key properties of the proposed mechanism remain under-explored.

In this paper, we analyze the mechanism 
of averaging preference vectors across a population, 
focusing on stability, strategic effects,
and implications for fairness.
While averaging mechanisms
have been extensively studied \citep{hurley2002combining,renault2005protecting,marchese2011strategy,renault2011assessing},
\citet{noothigattu2018voting}'s approach
warrants analysis for several reasons:
(i) here, individuals vote on the parameters of a ranking algorithm,
rather than directly on the outcomes of interest
and (ii) since only the ordering induced 
by the preference vectors matters,
only the direction of the aggregate vector is relevant.
We introduce a stylized model 
in which the population of interest consists 
of a (disadvantaged) minority group,
constituting an $\alpha < 0.5$ share of the population,
and an (advantaged) majority group constituting $1-\alpha$.
The goal is to determine how an autonomous vehicle should behave
by polling the population via \citet{noothigattu2018voting}'s
algorithm, where the problem setup is identical to that of the Moral Machine \cite{awad2018moral}.
Because we only
observe 
which alternative an individual prefers,
our analysis focuses on the fraction of cases
(among \emph{disputed cases})
where each group prevails.

To clarify the fairness properties of the mechanism,
we concentrate our analysis on an stylized setting 
in which within-group preferences are homogeneous.
We emphasize this assumption is only meant for simplifying the analysis and revealing fundamental limitations of the simple preference vector averaging mechanism.
The problems we identify here do not simply disappear 
in more complicated settings with within-group variation in preferences. 
Moreover, to isolate the role of the aggregation mechanism,
we assume participants
can directly report their preference vectors. 
With these assumptions in place, our analysis makes the following key observations:
(i) even when preferences are reported truthfully, 
the fraction of cases where the minority prevails 
is sub-proportionate (i.e., less than $\alpha$);
(ii) the degree of sub-proportionality grows more severe 
when the divergence between the two groups' preference vectors is large;
(iii) as with most averaging-based approaches,
    this mechanism is not strategy-proof;
(iv) whenever a pure strategy equilibrium exists,
    the {\adv} group prevails on 100\% of cases;
and (v) last but not least, while other stable and incentive-compatible mechanisms do exist 
(e.g., the randomized dictatorship model), 
they come with other fundamental shortcomings.

Several takeaways flow from our analysis.
First, this averaging of preference vectors 
is qualitatively different from averaging 
votes directly on outcomes of interest,
giving rise to instability 
and surprising strategic behavior.
Second, the degree of compromise 
among majority and minority demographics,
both under truthful and strategic settings,
is an important consideration when designing aggregation mechanisms 
to support participatory ML systems.
Finally, our work raises critical questions regarding the limitations of
simple voting methods to ensure stakeholders' participation in the design of high-stakes automated decision-making systems.
We hope that this work encourages the community
to reflect on the importance of addressing normative disagreements among stakeholders---not 
simply by passing them through an aggregation mechanism, 
but by effectively giving voice to the disadvantaged communities
and facilitating deliberations necessary to reach an acceptable outcome for all stakeholders.

\section{Related Work}
\label{sec:related}
Our work draws on the preference elicitation 
and computational social choice literatures.
We build most directly on a line of work
consisting of the Moral Machine \citep{awad2018moral}
and subsequently proposed procedures 
for inferring and aggregating preferences \citep{noothigattu2018voting}.
These studies inspired numerous follow-up articles, 
such as \citet{pugnetti2018customer},
who pose the same questions to Swiss vehicle customers;
\citet{wang2019privacy}, who modify the algorithms 
to support differential privacy; 
and \citet{kemmer2020enhancing}, who evaluate
various methods of aggregating crowdsourced judgments, including averaging.

\paragraph{Preference elicitation for participatory ML}
\citet{lee2019webuildai} helped nonprofit volunteers
build ML algorithms via pairwise comparisons 
and aggregating the resulting preferences 
using Borda Count voting. 
\citet{johnstonpreference} adopted participatory mechanisms 
to determine how to allocate COVID-19 triage supplies. 
\citet{freedman2020adapting} modified parameter weights 
in their kidney exchange linear programs to help break ties 
based on inferred participant preferences 
about who should receive kidney donations.
Still other works that learn fairness metrics from user input
assume a single participant 
or a group capable of coming to consensus, 
and therefore employ no aggregation mechanism 
\citep{ilvento2019metric, hiranandani2020quadratic, hiranandani2020fair}.
Note that, per \citet{chamberlin1985investigation},
that using Borda Count instead of averaging 
does not necessarily alleviate issues
related to strategic voting.
Similarly, a minimax group regret approach
as employed in \citet{johnstonpreference}
may also be vulnerable to subgroup strategic voting. 

\paragraph{Computational Social Choice  }
Like us, \citet{el2021strategyproofness} highlights 
the general susceptibility of averaging-based methods
to strategic voting and discusses alternative 
median-based mechanisms.
While they mention the Moral Machine as an example,
they do not analyze the particular mechanism
presented in \citep{noothigattu2018voting} 
or provide any of the insights 
about fairness and strategic concerns
presented in our work.
\citet{moulin1980strategy, conitzer2016rules, zhang2019better}
propose median-based algorithms for voting in a strategy-proof manner
in the context of crowdsourcing societal tradeoffs.
\citet{brill2015strategic} address strategic voting,
noting that median-based approaches 
are less susceptible to manipulation.
\citet{conitzer2015crowdsourcing} describes issues 
with crowdsourcing societal tradeoffs more generally. 
Concerning randomized dictatorship models,
\citet{gibbard1973manipulation, gibbard1977manipulation} 
introduce randomized solutions as ``unattractive'' 
yet strategy-proof approaches.
\citet{zeckhauser1973voting} similarly states 
that randomized dictatorship 
has some favorable characteristics
in that it forces voters 
to report their true preferences 
(i.e., is strategy-proof) 
and is \emph{probabilistically linear} 
(i.e., switching votes from one alternative to another 
only affects the selection probabilities 
of those two alternatives in a linear fashion).
Instead of proposing other aggregation mechanisms, both \citet{landemore2015deliberation} and \citet{pierson2017demographics}
argue that deliberation and discussion among participants
are procedures that, when used in conjunction with voting, can lead to better outcomes than voting as a standalone process.

\section{Our Problem Setup}
\label{sec:setup}
Consider a population consisting of two groups: 
{\A}, the majority, and {\D},
a minority constituting $\alpha < 0.5$ fraction of the population.
Each group is characterized by 
a true \emph{preference vector} 
$\vec{\theta}^*_i \in \mathbb{R}^d$
that determines the preferences of all members of that group over outcomes/alternatives.
(In the autonomous vehicle example, the outcome could be the individuals chosen to be saved in face of an unavoidable accident.)
Each alternative is represented by a feature vector 
$\vecl{x} \in \mathbb{R}^d$.
We assume alternatives
are drawn independently from $M$, 
a spherically symmetric distribution 
centered at the origin with radius 1.
Members of group $i$ prefer alternative $\vecl{x} \in \mathbb{R}^d$
to alternative $\vecl{y} \in \mathbb{R}^d$
whenever $\vec{\theta}^*_i . \vecl{x} \geq \vec{\theta}^*_i . \vecl{y}$, 
i.e., whenever $\vec{\theta}^*_i . (\vecl{x} - \vecl{y}) \geq 0$. So the true preference vector 
$\vec{\theta}^*_i \in \mathbb{R}^d$ 
of group $i \in \{\A,\D\}$ 
allows us to calculate their rankings 
over any set of alternatives.
See Appendix \ref{appendix:vector-example} for a concrete example of such vectors.
Throughout and unless otherwise specified, we assume $\vec{\theta}^*_A \neq \vec{\theta}^*_D$ (e.g., in the context of the Moral Machine experiment, $\vec{\theta}^*_A$ may reflect group A's preference for sacrificing pedestrians when an autonomous vehicle is faced with an accident, and $\vec{\theta}^*_D$ reflects group D's preference for sacrificing passengers).

Note that in enforcing within-group homogeneity of preferences in our model, we do not assert that groups are homogeneous in the real world. Rather, our aim is to elucidate whether this mechanism respects the preferences of minority groups. In other words, the homogeneity assumption allows us to make a clear analytic point rather than state a realistic or normatively desirable situation. Such simplifying assumptions are common in economics and theoretical computer science literature (e.g., \citet{costinot2007optimal, krishna2012voluntary}).

\paragraph{Formulation as a game}
\label{subsec:model}
We consider a two-player normal-form game
$G={(S_\A, u_\A), (S_\D, u_\D)}$ 
in which groups in the setup described above correspond to players, $\A$ and $\D$ (so player $\A$ and $\D$ have true preference vectors
$\vec{\theta}^*_\A$ and $\vec{\theta}^*_\D$, respectively). $S_i$ denotes player $i$'s strategy space, and $u_i$ their payoff/utility. More precisely, each player $i\in \{A,D\}$ strategically reports 
a preference vector $\vec{\theta}_i \in S_i$ 
(which may be different from $\vec{\theta}^*_i$), where 
$S_i$ consists of all $d$-dimensional vectors 
with Euclidean norm equal to 1. 
The payoff function for player $i$, $u_i: \mathbb{R}^d \times \mathbb{R}^d \longrightarrow \mathbb{R}$, 
is a function where $u_i(\vec{\theta}_\A, \vec{\theta}_\D)$ indicates the payoff of player $i \in \{\A,\D\}$ 
when players $\A$ and $\D$ report preference vectors
$\vec{\theta}_\A$ and $\vec{\theta}_\D$, respectively. 
We assume the payoff function for each player $i \in \{\A,\D\}$ 
captures the proximity of their true preference vector
$\vec{\theta}^*_i$ to the aggregate vector,
$\vec{\theta}_C$, obtained by a \emph{simple averaging} mechanism:
$\vec{\theta}_C = (\alpha \vec{\theta}_D + (1-\alpha) \vec{\theta}_A)/{\lVert \alpha \vec{\theta}_D + (1-\alpha) \vec{\theta}_A \rVert}$.\footnote{While only the direction of $\vec{\theta}_C$ matters, 
we normalize its length for mathematical convenience.}

Each player aims to maximize 
the fraction of decisions
that agree with their true preferences. 
As shown in Proposition~\ref{thm:agreement-prob}, achieving this goal requires 
that each player $i\in \{A,D\}$
reports a preference vector such that 
the resulting aggregate vector $\vec{\theta}_C$ 
is closest, as measured by cosine similarity,
to their true preference vector $\vec{\theta}^*_i$. 
To show this formally, we first define the notion of \emph{agreement} between two preference vectors.
\begin{definition}
Consider two preference vectors 
$\vec{\theta}_i$ and $\vec{\theta}_j \in \mathbb{R}^d$ and a spherically symmetric distribution $M$ over the space of all alternatives (i.e., $\mathbb{R}^d$).
The \emph{level of agreement} between $\vec{\theta}_i$ and $\vec{\theta}_j$, denoted by
$\rho(\vec{\theta}_i, \vec{\theta}_j)$, 
is the probability, 
over draws of pairs of alternatives from $M$, such that $\vec{\theta}_i$ and $\vec{\theta}_j$ 
rank the alternatives in the same order: 
\begin{equation*}
    \rho(\vec{\theta}_i, \vec{\theta}_j)=
    \mathbb{P}_{\vecl{x} \sim M, \vecl{y} \sim M}\left[\sign({\vec{\theta}_i} \cdot (\vecl{y} - \vecl{x})) = \sign({\vec{\theta}_j} \cdot (\vecl{y} - \vecl{x}))\right].
\end{equation*}
\end{definition}

\begin{proposition}
\label{thm:agreement-prob}
Suppose alternatives are sampled i.i.d. from a spherically symmetric distribution $M$ defined over $\mathbb{R}^d$.  
Then for any two preference vectors
$\vec{\theta}_i, \vec{\theta}_j \in \mathbb{R}^d$,
\begin{equation*}
    \rho(\vec{\theta}_i, \vec{\theta}_j) =
    \frac{\pi - \cos^{-1}{(\vec{\theta}_i \cdot \vec{\theta}_j})}{\pi}.
\end{equation*}
\end{proposition}
\begin{proof}

Note that the preference of a player $i$ 
over alternatives $\vecl{x}$ and $\vecl{y}$
depends only on the sign of
$\vec{\theta}_{i} \cdot (\vecl{y} - \vecl{x})$.
Because $\vecl{x}$ and $\vecl{y}$
are drawn independently from
a spherically symmetrical distribution $M$,
we can see, by symmetry,
that the difference vector, $\vecl{y} - \vecl{x}$,
can points in any direction from $[0,2\pi]$ 
with equal (uniform) probability. 

Any preference vector $\vec{\theta}$ 
defines two half-spaces over the vector
$\vecl{y}-\vecl{x} \in \mathbb{R}^d$.
By $H^{+}$, we denote the half-space 
in which $\vec{\theta} \cdot (\vecl{y}-\vecl{x}) > 0$
(and thus $\vecl{y}$ is preferred to $\vecl{x}$)
and by $H^{-}$, we denote the half-space 
in which  $\vec{\theta} \cdot (\vecl{y}-\vecl{x}) < 0$ 
(and thus $\vecl{x}$ is preferred to $\vecl{y}$).
Because the event where $\vecl{x} = \vecl{y}$ 
has $0$ measure, tie-breaking conventions 
will not impact our analysis.

Note that for any given player, 
the line separating 
$H^{+}$ from $H^{-}$ is 
perpendicular to $\vec{\theta}$
and passes through the origin.
Disagreements among 
$\vec{\theta}_i$ and $\vec{\theta}_j$
correspond to pairs of alternatives
such that either 
$\vecl{x} - \vecl{y} \in 
\{H_j^{-} \cap H_i^{+}\}$
or $\vecl{x} - \vecl{y} \in 
\{H_{i}^- \cap H_{j}^+\}$. So we have:
\begin{eqnarray*}
    \rho(\vec{\theta}_i, \vec{\theta}_j)&=&
    \mathbb{P}_{\vecl{x} \sim M, \vecl{y} \sim M}\left[\sign({\vec{\theta}_i} \cdot (\vecl{y} - \vecl{x})) = \sign({\vec{\theta}_j} \cdot (\vecl{y} - \vecl{x}))\right],\\
    &=& 1- \mathbb{P}_{\vecl{x} \sim M, \vecl{y} \sim M}\left[\vecl{x} - \vecl{y} \in H_j^{-} \cap H_i^{+} \text{ or } \vecl{x} - \vecl{y} \in H_j^{+} \cap H_i^{-}\right].
\end{eqnarray*}
Note that both $\{H_j^{-} \cap H_i^{+}\}$
and $\{H_i^{-} \cap H_j^{+}\}$
are cones whose vertices lie at the origin
and whose vertex angles are each given 
by $\measuredangle(\vec{\theta}_i, \vec{\theta}_j)$. See Figure \ref{fig:half-spaces}
for an illustration of these half-spaces and their intersections.

\begin{figure}
    \centering
    \includegraphics[width=0.32\textwidth]{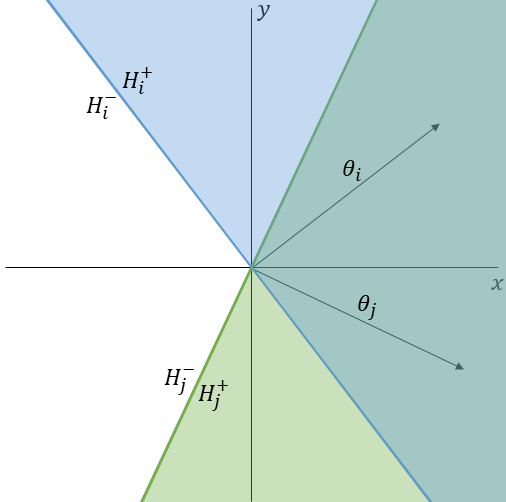}
    \caption{Half-spaces $H_i^+, H_i^-, H_j^+, H_j^-$ and intersections for preference vectors $\vec{\theta}_i$ and $\vec{\theta}_j$.}
    \label{fig:half-spaces}
    \vspace{-3mm}
\end{figure}

Because the alternatives, $\vecl{x},\vecl{y}$ are drawn i.i.d.
from a spherically symmetric distribution,
\begin{align*}
\mathbb{P}\left[\vecl{x} - \vecl{y} \in \{H_j^{-} \cap H_i^{+}\}\right]
&= \mathbb{P}\left[\vecl{x} - \vecl{y} \in
\{H_i^{-} \cap H_j^{+}\}\right]\\
&= \frac{\cos^{-1}(\vec{\theta}_i \cdot \vec{\theta}_j)}{2\pi}.    
\end{align*}

Moreover, because these regions of disagreement are disjoint,
the total probability of disagreement
$\vec{\theta}_i$,
$\vec{\theta}_j$ is given by the sum
of the probabilities of $\vecl{y} - \vecl{x}$
lying in either region of disagreement,
$\cos^{-1}{(\vec{\theta}_i \cdot \vec{\theta}_j)})/2\pi$.
Thus, the level of agreement is given by
$
  \frac{\pi - \cos^{-1}{(\vec{\theta}_i \cdot \vec{\theta}_j)}}{\pi}.
$
\end{proof}
Because the level of agreement of 
the aggregate decisions with player $i$
is monotonic in the cosine similarity 
between their true preferences $\vec{\theta}^*_i$ 
and the aggregate vector $\vec{\theta}_C$,
we can equivalently take the cosine similarity 
as the payoff of interest. 
More precisely, we can define $u_\A, u_\D$ as follows: for all $\vec{\theta}_\A \in S_\A$ 
and $\vec{\theta}_\D \in S_\D$,
$$u_\A(\vec{\theta}_\A, \vec{\theta}_\D) = \vec{\theta}_C \cdot \vec{\theta}^*_\A,$$
$$u_\D(\vec{\theta}_\A, \vec{\theta}_\D) = \vec{\theta}_C \cdot \vec{\theta}^*_\D.$$

We are interested in understanding 
pure strategy Nash equilibria of the above game. 
To define pure strategy Nash equilibria precisely, 
we need to first define the concept
\emph{best responses} to a given pure strategy. We say
a pure strategy $\vec{\theta}_i \in S_i$ 
is a \emph{best response} 
to $\vec{\theta}_{-i} \in S_{-i}$ in $G$ 
if for all $\vec{\hat{\theta}}_i \in  S_i$,
$$u_i(\vec{\theta}_i, \vec{\theta}_{-i}) \geq u_i(\vec{\hat{\theta}}_i, \vec{\theta}_{-i}).$$
Given a strategy $\vec{\theta}_{-i}$ for player ${-i}$, 
we will use the notation $BR_i(\vec{\theta}_{-i})$ 
to refer to the set of all pure best responses of $i$ to $\vec{\theta}_{-i}$.

\begin{definition}[Nash Equilibrium]
A strategy profile $(\vec{\theta}'_\A, \vec{\theta}'_\D)$ 
is a pure Nash Equilibrium for $G$ 
if $\vec{\theta}'_\A$ is a best-response 
to $\vec{\theta}'_\D$ and vice versa.
\end{definition}
From here on, we restrict our analysis to the case where $d=2$.
As we argue in Appendix \ref{appendix:two-dim},
it may suffice to consider 
the 2-dimensional plane containing the origin,
$\vec{\theta}^*_\A$, and $\vec{\theta}^*_\D$.

\section{Findings}
\label{sec:analysis}
We are now ready to introduce our primary findings.

\paragraph{Disproportionate majority voice via the averaging mechanisms}
Assume that all participants truthfully report their preferences. 
Even here, the averaging mechanism has some strange properties.
Notably, focusing on the fraction of all disputed cases 
where group {\D} prevails, specifically

\text{}
\begin{equation}
\label{eq:conditional-minority-agreement}
    \mathbb{P}[\textrm{sign}({\vec{\theta}_C} \cdot (\vecl{y} - \vecl{x})) = \textrm{sign}({\vec{\theta}^*_\D} \cdot (\vecl{y} - \vecl{x}))\: |\:  
    \textrm{sign}({\vec{\theta}^*_\A} \cdot (\vecl{y} - \vecl{x})) \neq \textrm{sign}({\vec{\theta}^*_\D} \cdot (\vecl{y} - \vecl{x}))]
\end{equation}
we find that this is  
sub-proportional in $\alpha$
(the pattern follows a sigmoid-like shape).
Moreover, the degree of sub-proportionality depends
on the angle between the true preference vectors because even though
$\vec{\theta}_C$ is a weighted sum of $\vec{\theta}^*_{\A}$ and $\vec{\theta}^*_{\D}$,
the direction of $\vec{\theta}_C$ (and in turn the levels of agreement with the 
{\dis} and {\adv} groups) depends on the directions of $\vec{\theta}^*_{\A}$ and $\vec{\theta}^*_{\D}$
in addition to $\alpha$.
This can be seen clearly in Figure \ref{fig:conditional-minority-consensus-cosine-sim} 
where the probability of group {\D} prevailing is computed
via Equation \ref{eq:conditional-minority-agreement}
under varying settings of their population share $\alpha$ 
and for various levels of agreement between the two groups (see Appendix \ref{appendix:verbose-eq}).
Note that sub-proportionality becomes more extreme 
as the angle between the true preference vectors increases.

\begin{figure}
    \centering
    \includegraphics[width=0.45\textwidth]{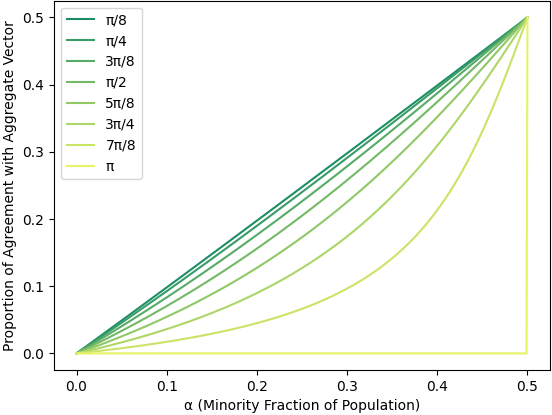}
    \caption{ 
    As $\alpha$ increases, 
    the probability that the aggregate agrees with the {\dis} group increases, 
    but this relationship is sub-proportional in $\alpha$ 
    and depends on $\measuredangle(\vec{\theta}^*_{\D}, \vec{\theta}^*_{\A})$.}
    \label{fig:conditional-minority-consensus-cosine-sim}
    \vspace{-3mm}
\end{figure}

\paragraph{Majority Group Can Create Aggregate Vector In Any Direction}
Next, we address the setting where groups can report their preferences strategically.
First, we find that for any preference vector $\vec{\theta}_\D$ 
reported by the minority group,
the {\adv} group can always choose some preference vector to report
such that the aggregate vector is identical to their true preferences 
$\vec{\theta}_C = \vec{\theta}^*_\A$. 
This implies that if a pure strategy Nash equilibrium exists,
group {\A} always gets their way. (In our running example, this implies that the autonomous vehicle always operates in accordance with group $A$'s preferences.)
\begin{lemma}
\label{thm:any-consensus}
For any fixed vector $\vec{\theta}_\D$ played by the minority group, 
the majority group can report a vector $\vec{\theta}_\A$,
such that $\vec{\theta}_C = \vec{\theta}^*_\A$.
\end{lemma}
\begin{proof}[Proof sketch]
In order to ensure $\vec{\theta}^*_\A = \vec{\theta}_C$, player $\A$ must report $\vec{\theta}_\A$ such that $\frac{\alpha \vec{\theta}_D + (1-\alpha) \vec{\theta}_A}{\lVert \alpha \vec{\theta}_D + (1-\alpha) \vec{\theta}_A \rVert} = \vec{\theta}^*_\A$.
For any vector $\vec{\theta}_\D$, it is easy to see that the above equation is equivalent to:
\begin{equation}
\label{eq:maj-group-any-direction}
     \vec{\theta}_\A = 
     \frac{\bigg[\alpha (\vec{\theta}_\D \cdot \vec{\theta}^*_\A) + \sqrt{\alpha^2 (\vec{\theta}_\D \cdot \vec{\theta}^*_\A)^2 - 2\alpha + 1}\bigg]\vec{\theta}^*_\A - \alpha \vec{\theta}_\D}{(1 - \alpha)}.
\end{equation}
(See Appendix \ref{appendix:maj-agg-extended-pf} for the full derivation.)
\end{proof}

\subsection{Conditions for Pure Strategy Nash Equilibrium}
In this section, we present a necessary condition for the existence of a pure Nash Equilibrium in the above game. 
First, we derive the maximum amount by which the {\dis} group can pull the aggregate vector based on their relative population size by reporting their preference strategically.

\begin{lemma}
\label{lem:minority-ort}
Consider a fixed reported {\adv} group vector $\vec{\theta}_\A$. Suppose the {\dis} group reports $\vec{\theta}_\D$ to yield an aggregate vector $\vec{\theta}_C$. Then for any $\vec{\theta}_\D \in S_\D$,
\begin{equation}
    \measuredangle \left(\vec{\theta}_C, \vec{\theta}_\A \right) \leq \sin^{-1}{\left(\frac{\alpha}{1-\alpha}\right)}.
\end{equation}
The equality occurs if and only if $\vec{\theta}_\D$ is orthogonal to $\vec{\theta}_C$.
\end{lemma}
\begin{proof}[Proof sketch]
It is easy to see that $\vec{\theta}_\D$ yields an aggregate vector $\vec{\theta}_C$ such that $\measuredangle \left(\vec{\theta}_C, \vec{\theta}_\A \right) = \sin^{-1}{\left(\frac{\alpha \sin{(\measuredangle \left(\vec{\theta}_C, \vec{\theta}_\D \right))}}{1-\alpha}\right)}$. Note that because $\sin{(\measuredangle \left(\vec{\theta}_C, \vec{\theta}_\D \right))} \leq 1$, $\left(\frac{\alpha \sin{(\measuredangle \left(\vec{\theta}_C, \vec{\theta}_\D \right))}}{1-\alpha}\right) \leq \left(\frac{\alpha}{1-\alpha}\right)$. Moreover, $\sin^{-1}(.)$ is a monotonic function in $[-1,1]$. Therefore, $\measuredangle \left(\vec{\theta}_C, \vec{\theta}_\A \right) = \sin^{-1}{\left(\frac{\alpha \sin{(\measuredangle \left(\vec{\theta}_C, \vec{\theta}_\D \right))}}{1-\alpha}\right)} \leq \sin^{-1}{\left(\frac{\alpha}{1-\alpha}\right)}$. Additionally, $\measuredangle \left(\vec{\theta}_C, \vec{\theta}_\A \right) = \sin^{-1}{\left(\frac{\alpha}{1-\alpha}\right)}$ if and only if $\sin{(\measuredangle \left(\vec{\theta}_C, \vec{\theta}_\D \right))} = 1$. This in turn happens if and only if $\vec{\theta}_C \perp \vec{\theta}_\D$.

(See Appendix \ref{appendix:dist-deriv} for the full derivation.) 
\end{proof}

These two lemmas yield the subsequent three lemmas, which are proved in Appendix \ref{appendix:corr-proofs}:

First, if a pure strategy Nash equilibrium exists, the {\adv} group can always report its preference vector such that the aggregate matches their true preferences.
\begin{lemma}
\label{cor:maj-consensus}
Consider the game, $G$, described in Section \ref{subsec:model}. If $(\vec{\theta}'_\D, \vec{\theta}'_\A)$ is a pure strategy Nash equilibrium for $G$, then $\vec{\theta}_C = \vec{\theta}^*_{\A}$. 
\end{lemma}

Second, we derive an upper bound on the angle between the aggregate and {\dis} group's true preferences.
\begin{lemma}
\label{cor:min-ub}
For any $\vec{\theta}_{\A} \in S_{\A}$, there exists $\vec{\theta}_\D \in S_{\D}$ such that $\measuredangle \left(\vec{\theta}_C, \vec{\theta}^*_{\D}\right) \leq \pi - \sin^{-1}\left(\frac{\alpha}{1-\alpha}\right)$.    
\end{lemma}

Lastly, in a pure strategy Nash equilibrium, the {\dis} group's best response is orthogonal to the aggregate.
\begin{lemma}
\label{lem:min-orth-eq}
In any pure strategy Nash equilibrium defined by best responses $\vec{\theta}'_{\A}$ and $\vec{\theta}'_{\D}$, $\vec{\theta}'_{\D} \perp \vec{\theta}_C$, and $\measuredangle \left(\vec{\theta}_C, \vec{\theta}'_\A \right) = \sin^{-1}{\left(\frac{\alpha}{1-\alpha}\right)}$.
\end{lemma}

Now, we show that a pure strategy Nash equilibrium does not always exist.

\begin{theorem}
\label{thm:eq-conditions}
For $G$ to have a pure strategy Nash equilibrium, it must be the case that
\begin{equation*}
 \measuredangle \left(\vec{\theta}^*_\A, \vec{\theta}^*_\D \right) < \pi - \sin^{-1}{\left(\frac{\alpha}{1-\alpha}\right)}.
\end{equation*}
\end{theorem}
\begin{proof}
Suppose, via contradiction, there exists a pure strategy Nash equilibrium and $\measuredangle \left(\vec{\theta}^*_\A, \vec{\theta}^*_\D \right) \geq \pi - \sin^{-1}{\left(\frac{\alpha}{1-\alpha}\right)}$. Two cases are possible:
\begin{enumerate}
    \item $\measuredangle \left(\vec{\theta}^*_\A, \vec{\theta}^*_\D \right) > \pi - \sin^{-1}{\left(\frac{\alpha}{1-\alpha}\right)}$;
    \item $\measuredangle \left(\vec{\theta}^*_\A, \vec{\theta}^*_\D \right) = \pi - \sin^{-1}{\left(\frac{\alpha}{1-\alpha}\right)}$.
\end{enumerate}

First, consider the case where $\measuredangle \left(\vec{\theta}^*_\A, \vec{\theta}^*_\D \right) > \pi - \sin^{-1}{\left(\frac{\alpha}{1-\alpha}\right)}$. By Lemma \ref{cor:maj-consensus}, $\vec{\theta}^*_\A = \vec{\theta}_C$ in this equilibrium. However, by Lemma \ref{cor:min-ub}, $\measuredangle \left(\vec{\theta}_C, \vec{\theta}^*_\D \right) \leq \pi - \sin^{-1}{\left(\frac{\alpha}{1-\alpha}\right)}$ given any response from the {\adv} group. This means the {\dis} group will report a best response $\vec{\theta}'_\D$ that does not yield $\vec{\theta}^*_{\A}$. This in turn raises a contradiction and indicates that such an equilibrium cannot exist if $\measuredangle \left(\vec{\theta}^*_\A, \vec{\theta}^*_\D \right) > \pi - \sin^{-1}{\left(\frac{\alpha}{1-\alpha}\right)}$.

Second, consider the case where $\measuredangle \left(\vec{\theta}^*_\A, \vec{\theta}^*_\D \right) = \pi - \sin^{-1}{\left(\frac{\alpha}{1-\alpha}\right)}$. By Lemma \ref{cor:maj-consensus}, $\vec{\theta}^*_\A = \vec{\theta}_C$ in this equilibrium. Because $\vec{\theta}^*_\A = \vec{\theta}_C$, by Lemma \ref{lem:min-orth-eq}, $\measuredangle \left(\vec{\theta}'_\A, \vec{\theta}^*_\A \right) = \left(\vec{\theta}'_\A, \vec{\theta}_C \right) = \sin^{-1}{\left(\frac{\alpha}{1-\alpha}\right)}$. Therefore $\measuredangle \left(\vec{\theta}'_{\A}, \vec{\theta}^*_{\D}\right) = \measuredangle \left(\vec{\theta}'_\A, \vec{\theta}^*_\A \right) + \measuredangle \left(\vec{\theta}^*_\A, \vec{\theta}^*_\D \right) = \sin^{-1}{\left(\frac{\alpha}{1-\alpha}\right)} + \pi - \sin^{-1}{\left(\frac{\alpha}{1-\alpha}\right)} = \pi$, and $\vec{\theta}'_{\A}$ is diametrically opposed to $\vec{\theta}^*_{\D}$. Note that there are two best responses for the {\dis} group, both of which pull the aggregate (one by $\sin^{-1}\left(\frac{\alpha}{1-\alpha}\right)$, the other by $-\sin^{-1}\left(\frac{\alpha}{1-\alpha}\right)$) towards $\vec{\theta}^*_\D$ to yield $\measuredangle \left(\vec{\theta}_C, \vec{\theta}^*_\D \right) = \pi - \sin^{-1}{\left(\frac{\alpha}{1-\alpha}\right)}$. However, only \textit{one} yields $\vec{\theta}_C = \vec{\theta}^*_{\A}$. This is a contradiction because $\vec{\theta}_C$ does not necessarily match $\vec{\theta}^*_\A$. Therefore, this is not a pure strategy Nash equilibrium.

Thus, for $G$ to have a pure strategy Nash equilibrium, $\measuredangle \left(\vec{\theta}^*_\A, \vec{\theta}^*_\D \right) < \pi - \sin^{-1}{\left(\frac{\alpha}{1-\alpha}\right)}$.
\end{proof}

Theorem \ref{thm:eq-conditions} reveals that equilibrium does not exist if (i) the groups are close to diametric opposition (in terms of their preference vectors), and (ii)
the groups are close in size.

\subsection{Form of Pure Strategy Nash Equilibrium}

Whenever necessary conditions outlined in Theorem \ref{thm:eq-conditions} are met, the equilibrium takes a certain form. Theorem \ref{thm:nondegen-eq} specifies this form exactly and proves the conditions are also sufficient. 
We denote
$R = \big(\begin{smallmatrix}
  0 & -1\\
  1 & 0
\end{smallmatrix}\big)$ as the $\frac{\pi}{2}$ radians rotation matrix.

\begin{theorem}
\label{thm:nondegen-eq}
If $\measuredangle \left(\vec{\theta}^*_\A, \vec{\theta}^*_\D \right) < \pi - \sin^{-1}{\left(\frac{\alpha}{1-\alpha}\right)}$, then there exists a pure strategy Nash equilibrium. Specifically, ($\vec{\theta}'_\A, \vec{\theta}'_\D$) form a pure strategy Nash equilibrium iff:
\begin{enumerate}
    
    \item \label{it:best-response-vals} $\vec{\theta}'_\A$ and $\vec{\theta}'_\D$ take the following form:
    \begin{align*}
        \vec{\theta}'_\A &= \frac{\big[(\sqrt{1 - 2\alpha})I - \alpha \sign(\vec{\theta}^*_\D \cdot R \vec{\theta}^*_{\A})R\big]\vec{\theta}^*_\A}{1 - \alpha},
        \\
        \vec{\theta}'_\D &= \sign(\vec{\theta}^*_\D \cdot R \vec{\theta}^*_{\A})R \vec{\theta}^*_{\A}.
    \end{align*}
    \item $\measuredangle \left(\vec{\theta}'_\A, \vec{\theta}'_\D \right) = \sin^{-1}{\left(\frac{\alpha}{1-\alpha}\right)} + \frac{\pi}{2}$. That is, players' equilibrium strategies will always point in opposing directions.
\end{enumerate}
\end{theorem}
\begin{proof}
Because $\measuredangle \left(\vec{\theta}^*_\A, \vec{\theta}^*_\D \right) < \pi - \sin^{-1}{\left(\frac{\alpha}{1-\alpha}\right)}$, the {\dis} group will always report the unique vector $\vec{\theta}'_\D$ orthogonal to $\vec{\theta}_C$ (as per Lemma \ref{lem:min-orth-eq}) that maximizes agreement between $\vec{\theta}_C$ and $\vec{\theta}^*_\D$. By Lemma \ref{cor:maj-consensus}, $\vec{\theta}_C = \vec{\theta}^*_\A$ at equilibrium, meaning that $\vec{\theta}'_\D$ is orthogonal to $\vec{\theta}^*_\A$. $\vec{\theta}'_\D$ is therefore either $R \vec{\theta}^*_\A$ or $-R \vec{\theta}^*_\A$.

If $\vec{\theta}^*_\D \cdot R \vec{\theta}^*_{\A} > 0$, then the {\dis} group should report $\vec{\theta}'_\D = R \vec{\theta}^*_\A$ to maximize their utility. Otherwise, the {\dis} group should report $\vec{\theta}'_\D = -R \vec{\theta}^*_\A$. In either case, $\vec{\theta}'_\D = \sign(\vec{\theta}^*_\D \cdot R \vec{\theta}^*_{\A})R \vec{\theta}^*_\A$. (Note that $\sign(\vec{\theta}^*_\D \cdot R \theta^*_{\A}) = 0$ if and only if the two groups are in total agreement or are diametrically opposed. The proof preconditions that the groups disagree and that $\measuredangle \left(\vec{\theta}^*_\A, \vec{\theta}^*_\D \right) < \pi - \sin^{-1}{\left(\frac{\alpha}{1-\alpha}\right)}$ prevent this result from occurring.)

Using Equation \ref{eq:maj-group-any-direction} with $\vec{\theta}'_\D$ and $\vec{\theta}^*_\A$ yields $\vec{\theta}'_\A$ and $\vec{\theta}'_\D$ in item 1. Appendix \ref{appendix:eq-alg} shows the steps involved in this process. Moreover, item 2 follows from Lemma \ref{lem:min-orth-eq} because $\measuredangle \left(\vec{\theta}'_\A, \vec{\theta}'_\D \right) = \measuredangle \left(\vec{\theta}'_\A, \vec{\theta}'_C \right) + \measuredangle \left(\vec{\theta}'_C, \vec{\theta}'_\D \right) = \sin^{-1}{\left(\frac{\alpha}{1-\alpha}\right)} + \frac{\pi}{2}$. See Figure \ref{fig:eq-result} for an illustration of all items.
\end{proof}

\begin{figure}[H]
    \centering
    \includegraphics[width=0.34\textwidth]{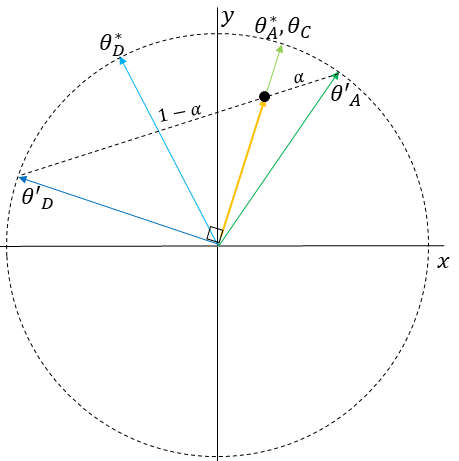}
    
    \caption{Example pure strategy Nash equilibrium with best responses $\vec{\theta}'_{\A}$ and $\vec{\theta}'_{\D}$ given minority population fraction $\alpha$, and true preference vectors $\vec{\theta}^*_\A$, and $\vec{\theta}^*_\D$. Note that $\vec{\theta}'_{\D} \perp \vec{\theta}_{C}$ and the aggregate vector $\vec{\theta}_C$ points in the same direction as $\vec{\theta}^*_{\A}$.}
    \label{fig:eq-result}
    \vspace{-5mm}
\end{figure}

\section{Discussion}
\label{sec:discussion}
Our analysis of the
averaging mechanism shows the following:
(a) even when groups are truthful,
concessions to the {\dis} group 
are less than proportional 
to their share of the population;
(b) this sub-proportionality 
depends on the cosine similarity 
between the true preference vectors;
and (c) if participants respond strategically,
 tyranny of the majority results 
whenever an equilibrium exists. 
We now build on this foundation by exploring:
(i) how our results obtain 
even without intragroup collusion;
(ii) the benefits and pitfalls 
of other aggregation mechanisms, 
(iii) more general implications 
for computational social choice 
and participatory ML algorithms.

\xhdr{No intra-group collusion required for results}
\label{subsec:no-collusion}
All of our analysis up to this point utilized a setup 
in which each group acts as a monolithic player.
However, we now informally argue that collusion is \emph{not} required.
Namely, as long as individuals are aware of the aggregate vector's direction and the relative sizes of each group, no collusion is necessary.
\begin{pfsketch}
Suppose colluding and not colluding result in different aggregate vectors.
Decomposing each aggregate vector into contributions from the {\dis} and {\adv}
group members (i.e. for a given group, summing the individual preference vectors of
members from that group and dividing by the number of members)
will result in different contributing vectors from each group. 
The contributing vectors that stem from collusion are guaranteed to maximize
the utility of each group. However, this means at least one of the contributing vectors
in the other set (resulting from independent contributions) does \textit{not} maximize
the utility of its corresponding group because it does not match the vector for that
same group via collusion. 
However, this raises a contradiction because maximizing group utility 
involves maximizing each group member's individual utility, 
regardless of whether collusion is involved. 
If maximal group utility is not achieved, 
at least one member is not maximizing their own utility. 
This cannot possibly happen because we assume 
that all individuals are rational actors 
seeking to maximize their utility, 
so this is an impossible case.
Therefore, no collusion is required to obtain
the results described in previous subsections.
\end{pfsketch}

\xhdr{Median-based approaches: geometric medians and coordinate-wise medians} \citet{el2021strategyproofness} mentions 
that average-based aggregation mechanisms
are susceptible to manipulation even by an individual.
In response, they describe two types 
of median-based aggregation approaches 
that generalize to $d$ dimensions: 
geometric medians and coordinate-wise medians.
Given a set of preference vectors, 
the geometric median is a vector that minimizes 
the Euclidean distance to each vector in the set, 
and the coordinate-wise median is a vector 
whose $i$th component is the one-dimensional median 
computed from the $i$th components of all provided vectors. 
The main result of \citet{el2021strategyproofness} 
is that the geometric median is not generally strategy-proof,
and they point to existing work \citep{sui2015approximately}
that proves the coordinate-wise median 
is strategy-proof but not group strategy-proof. 
In our setup, the coordinate-wise median 
is both incentive-compatible and stable,
but it does not alleviate the dominance of the majority's voice. In fact, 
this mechanism might be considered \emph{worse},
because even when both groups report preferences truthfully,
the majority group always prevails.\footnote{It may also possible to transform the 
two-dimensional problem on a circle 
into a one-dimensional problem on a line, 
where the single-peaked preferences result 
of \citet{moulin1980strategy} can be applied.}

\xhdr{Randomized Dictatorship}
In contrast to median-based approaches, 
the method of randomized dictatorship solves 
all problems related to averaging aggregation.
In this method, the preferences of an individual
selected at random from the population 
are used directly as the aggregated result. 
As reported in \citet{zeckhauser1973voting},
it is strategy-proof and \emph{probabilistically linear}.
In the setting we consider, participants 
would not be incentivized to lie 
(because any participant's reported preferences could be applied to everyone), 
and concessions to the {\dis} group are proportional instead of sub-proportional 
(because the {\dis} group preferences will be selected with probability $\alpha$ 
and the {\adv} group preferences will be selected with probability $1 - \alpha$). 
Thus, the randomized dictatorship mechanism is both 
incentive-compatible and proportional. 

In the economics literature \citep{gibbard1973manipulation, gibbard1977manipulation},
researchers have contemplated such mechanisms 
but have argued that they are ``unattractive'' 
despite being strategy-proof 
because they ``[leave] too much to chance'' and ignore
input from all individuals except the one selected at random.
\citet{zeckhauser1973voting} also posits 
that using the preferences of one over many 
may not be appropriate when dealing 
with ``momentous social decisions''. 
Deciding the ethics of autonomous vehicles 
may be one such decision. 
\citet{conitzer2015crowdsourcing} additionally notes that one 
may be more confident in their colleagues' preferences
than those of a random member of the population,
so it may not meet requirements of procedural fairness 
despite giving proportional voice to the minority in expectation.

\xhdr{Additional Considerations}
\citet{zeckhauser1973voting} proves that 
``No voting system that relies on individuals' self-interested balloting 
can guarantee a nondictatorial outcome 
that is Pareto-optimal.'' 
In their conclusion, they note 
that while this is a pessimistic result, 
``perhaps we should not ask despairingly, Where
do we go from here?, but rather inquire, Have we
been employing the appropriate mind-set all along?''
In other words, the correct approach may not be to build 
a ``one-size-fits-all'' strategy-proof voting scheme
but rather something that holistically considers all preferences 
regarding a given social issue and surrounding context. 
\citet{conitzer2015crowdsourcing} notes that 
participation requires context and locality,
both of which are lost when we crowdsource moral dilemmas (as in \citet{awad2018moral}). 
\citet{pugnetti2018customer} echoes this message; they find that Swiss residents had different preferences 
about autonomous vehicle ethics relative to those of other countries. 
Additionally, \citet{conitzer2015crowdsourcing} calls attention 
to the importance of the featurization process of the alternatives. 
If these alternatives are not represented properly,
the elicitation and aggregation processes
cannot hope to arrive at a result
that accurately reflects participants' true judgments.
Moreover, \citet{landemore2015deliberation} and \citet{pierson2017demographics}
suggest deliberation and discussion used alongside voting can lead to better overall outcomes than voting alone.

\xhdr{Conclusion}
The central impulse of participatory machine learning is
to integrate input from various stakeholders directly
into the process of developing machine learning systems.
For all its promise, participatory ML
also raises challenging questions
about the precise form that 
such an integration should take. 
Our work highlights some of the challenges 
of designing such mechanisms,
especially when individuals may hold 
radically different values and act strategically. 
While we draw heavily on the previous literature 
in economics and computational social science,
our work also reveals that some of the ways of combining
preferences that seem natural from a machine learning perspective
can be unstable and lead to strange strategic behavior
in ways that do not map so neatly 
onto known analyses. 
Surprisingly, while preference elicitation 
has gained considerable attention 
in the participatory ML literature,
few papers address cases in which 
stakeholders hold genuinely conflicting values. 
While some part of this work going forward 
will surely be to study such mechanisms formally,
we also stress that better mechanism design 
is no panacea for reconciling conflicting values 
in the real world.
Beyond theoretical analysis, participatory ML systems
may also require channels by which communication,
debate, and reconciliation of competing values
could potentially take place.


%

\section*{Acknowledgements}
Authors acknowledges support from NSF (IIS2040929) and PwC (through the Digital Transformation
and Innovation Center at CMU). Any opinions, findings, conclusions, or recommendations expressed in this material are those of the authors and do not reflect the views of the National Science Foundation and other funding agencies.
The first author was partially supported by a GEM Fellowship and an ARCS Scholarship.


\bibliographystyle{apa-good}
\bibliography{refs}

\begin{thebibliography}{35}
\expandafter\ifx\csname natexlab\endcsname\relax\def\natexlab#1{#1}\fi
\expandafter\ifx\csname url\endcsname\relax
  \def\url#1{{\tt #1}}\fi
\expandafter\ifx\csname urlprefix\endcsname\relax\def\urlprefix{URL }\fi

\bibitem[{Angwin et~al.(2016)Angwin, Larson, Mattu, \&
  Kirchner}]{angwin2016machine}
Angwin, J., Larson, J., Mattu, S., \& Kirchner, L. (2016).
\newblock Machine bias: There’s software used across the country to predict
  future criminals. and it’s biased against blacks.
\newblock {\em ProPublica\/}.

\bibitem[{Awad et~al.(2018)Awad, Dsouza, Kim, Schulz, Henrich, Shariff,
  Bonnefon, \& Rahwan}]{awad2018moral}
Awad, E., Dsouza, S., Kim, R., Schulz, J., Henrich, J., Shariff, A., Bonnefon,
  J.-F., \& Rahwan, I. (2018).
\newblock The moral machine experiment.
\newblock {\em Nature\/}, {\em 563\/}(7729), 59--64.

\bibitem[{Brill \& Conitzer(2015)}]{brill2015strategic}
Brill, M., \& Conitzer, V. (2015).
\newblock Strategic voting and strategic candidacy.
\newblock In {\em Association for the Advancement of Artificial Intelligence
  (AAAI)\/}.

\bibitem[{Buolamwini \& Gebru(2018)}]{buolamwini2018gender}
Buolamwini, J., \& Gebru, T. (2018).
\newblock Gender shades: Intersectional accuracy disparities in commercial
  gender classification.
\newblock In {\em ACM Conference on Fairness, Accountability and Transparency
  (FAccT)\/}.

\bibitem[{Chamberlin(1985)}]{chamberlin1985investigation}
Chamberlin, J.~R. (1985).
\newblock An investigation into the relative manipulability of four voting
  systems.
\newblock {\em Behavioral Science\/}, {\em 30\/}(4), 195--203.

\bibitem[{Conitzer et~al.(2015)Conitzer, Brill, \&
  Freeman}]{conitzer2015crowdsourcing}
Conitzer, V., Brill, M., \& Freeman, R. (2015).
\newblock Crowdsourcing societal tradeoffs.
\newblock In {\em International Conference on Autonomous Agents and Multiagent
  Systems (AAMAS)\/}.

\bibitem[{Conitzer et~al.(2016)Conitzer, Freeman, Brill, \&
  Li}]{conitzer2016rules}
Conitzer, V., Freeman, R., Brill, M., \& Li, Y. (2016).
\newblock Rules for choosing societal tradeoffs.
\newblock In {\em Association for the Advancement of Artificial Intelligence
  (AAAI)\/}.

\bibitem[{Costinot \& Kartik(2007)}]{costinot2007optimal}
Costinot, A., \& Kartik, N. (2007).
\newblock On optimal voting rules under homogeneous preferences.
\newblock {\em V Manuscript, Department of Economics, University of California
  San Diego\/}.

\bibitem[{El-Mhamdi et~al.(2021)El-Mhamdi, Farhadkhani, Guerraoui, \&
  Hoang}]{el2021strategyproofness}
El-Mhamdi, E.-M., Farhadkhani, S., Guerraoui, R., \& Hoang, L.-N. (2021).
\newblock On the strategyproofness of the geometric median.
\newblock {\em arXiv preprint arXiv:2106.02394\/}.

\bibitem[{Freedman et~al.(2020)Freedman, Borg, Sinnott-Armstrong, Dickerson, \&
  Conitzer}]{freedman2020adapting}
Freedman, R., Borg, J.~S., Sinnott-Armstrong, W., Dickerson, J.~P., \&
  Conitzer, V. (2020).
\newblock Adapting a kidney exchange algorithm to align with human values.
\newblock {\em Artificial Intelligence\/}, {\em 283\/}, 103261.

\bibitem[{Gibbard(1973)}]{gibbard1973manipulation}
Gibbard, A. (1973).
\newblock Manipulation of voting schemes: a general result.
\newblock {\em Econometrica: Journal of the Econometric Society\/}, (pp.
  587--601).

\bibitem[{Gibbard(1977)}]{gibbard1977manipulation}
Gibbard, A. (1977).
\newblock Manipulation of schemes that mix voting with chance.
\newblock {\em Econometrica: Journal of the Econometric Society\/}, (pp.
  665--681).

\bibitem[{Hiranandani et~al.(2020{\natexlab{a}})Hiranandani, Mathur,
  Narasimhan, \& Koyejo}]{hiranandani2020quadratic}
Hiranandani, G., Mathur, J., Narasimhan, H., \& Koyejo, O.
  (2020{\natexlab{a}}).
\newblock Quadratic metric elicitation for fairness and beyond.
\newblock {\em arXiv preprint arXiv:2011.01516\/}.

\bibitem[{Hiranandani et~al.(2020{\natexlab{b}})Hiranandani, Narasimhan, \&
  Koyejo}]{hiranandani2020fair}
Hiranandani, G., Narasimhan, H., \& Koyejo, S. (2020{\natexlab{b}}).
\newblock Fair performance metric elicitation.
\newblock In {\em Advances in Neural Information Processing Systems
  (NeurIPS)\/}.

\bibitem[{Hurley \& Lior(2002)}]{hurley2002combining}
Hurley, W., \& Lior, D. (2002).
\newblock Combining expert judgment: On the performance of trimmed mean vote
  aggregation procedures in the presence of strategic voting.
\newblock {\em European Journal of Operational Research\/}, {\em 140\/}(1),
  142--147.

\bibitem[{Ilvento(2020)}]{ilvento2019metric}
Ilvento, C. (2020).
\newblock Metric learning for individual fairness.
\newblock In {\em Foundations of Responsible Computing (FORC)\/}.

\bibitem[{Johnston et~al.(2020)Johnston, Blessenohl, \&
  Vayanos}]{johnstonpreference}
Johnston, C.~M., Blessenohl, S., \& Vayanos, P. (2020).
\newblock Preference elicitation and aggregation to aid with patient triage
  during the covid-19 pandemic.
\newblock In {\em ICML Workshop on Participatory Approaches to Machine
  Learning\/}.

\bibitem[{Jung et~al.(2021)Jung, Kearns, Neel, Roth, Stapleton, \&
  Wu}]{jung2019algorithmic}
Jung, C., Kearns, M., Neel, S., Roth, A., Stapleton, L., \& Wu, Z.~S. (2021).
\newblock An algorithmic framework for fairness elicitation.
\newblock In {\em Foundations of Responsible Computing (FORC)\/}.

\bibitem[{Kemmer et~al.(2020)Kemmer, Yoo, Escobedo, \&
  Maciejewski}]{kemmer2020enhancing}
Kemmer, R., Yoo, Y., Escobedo, A., \& Maciejewski, R. (2020).
\newblock Enhancing collective estimates by aggregating cardinal and ordinal
  inputs.
\newblock In {\em AAAI Conference on Human Computation and Crowdsourcing
  (HCOMP)\/}.

\bibitem[{Krishna \& Morgan(2012)}]{krishna2012voluntary}
Krishna, V., \& Morgan, J. (2012).
\newblock Voluntary voting: Costs and benefits.
\newblock {\em Journal of Economic Theory\/}, {\em 147\/}(6), 2083--2123.

\bibitem[{Landemore \& Page(2015)}]{landemore2015deliberation}
Landemore, H., \& Page, S.~E. (2015).
\newblock Deliberation and disagreement: Problem solving, prediction, and
  positive dissensus.
\newblock {\em Politics, philosophy \& economics\/}, {\em 14\/}(3), 229--254.

\bibitem[{Lee et~al.(2019)Lee, Kusbit, Kahng, Kim, Yuan, Chan, See,
  Noothigattu, Lee, Psomas et~al.}]{lee2019webuildai}
Lee, M.~K., Kusbit, D., Kahng, A., Kim, J.~T., Yuan, X., Chan, A., See, D.,
  Noothigattu, R., Lee, S., Psomas, A., et~al. (2019).
\newblock Webuildai: Participatory framework for algorithmic governance.
\newblock In {\em ACM Conference on Computer- Supported Cooperative Work And
  Social Computing (CSCW)\/}.

\bibitem[{Marchese \& Montefiori(2011)}]{marchese2011strategy}
Marchese, C., \& Montefiori, M. (2011).
\newblock Strategy versus sincerity in mean voting.
\newblock {\em Journal of Economic Psychology\/}, {\em 32\/}(1), 93--102.

\bibitem[{Moulin(1980)}]{moulin1980strategy}
Moulin, H. (1980).
\newblock On strategy-proofness and single peakedness.
\newblock {\em Public Choice\/}, {\em 35\/}(4), 437--455.

\bibitem[{Noothigattu et~al.(2018)Noothigattu, Gaikwad, Awad, Dsouza, Rahwan,
  Ravikumar, \& Procaccia}]{noothigattu2018voting}
Noothigattu, R., Gaikwad, S., Awad, E., Dsouza, S., Rahwan, I., Ravikumar, P.,
  \& Procaccia, A. (2018).
\newblock A voting-based system for ethical decision making.
\newblock In {\em Association for the Advancement of Artificial Intelligence
  (AAAI)\/}.

\bibitem[{Obermeyer et~al.(2019)Obermeyer, Powers, Vogeli, \&
  Mullainathan}]{obermeyer2019dissecting}
Obermeyer, Z., Powers, B., Vogeli, C., \& Mullainathan, S. (2019).
\newblock Dissecting racial bias in an algorithm used to manage the health of
  populations.
\newblock {\em Science\/}, {\em 366\/}(6464), 447--453.

\bibitem[{Pierson(2017)}]{pierson2017demographics}
Pierson, E. (2017).
\newblock Demographics and discussion influence views on algorithmic fairness.
\newblock {\em arXiv preprint arXiv:1712.09124\/}.

\bibitem[{Pugnetti \& Schl{\"a}pfer(2018)}]{pugnetti2018customer}
Pugnetti, C., \& Schl{\"a}pfer, R. (2018).
\newblock Customer preferences and implicit tradeoffs in accident scenarios for
  self-driving vehicle algorithms.
\newblock {\em Journal of Risk and Financial Management\/}, {\em 11\/}(2), 28.

\bibitem[{Renault \& Trannoy(2005)}]{renault2005protecting}
Renault, R., \& Trannoy, A. (2005).
\newblock Protecting minorities through the average voting rule.
\newblock {\em Journal of Public Economic Theory\/}, {\em 7\/}(2), 169--199.

\bibitem[{Renault \& Trannoy(2011)}]{renault2011assessing}
Renault, R., \& Trannoy, A. (2011).
\newblock Assessing the extent of strategic manipulation: the average vote
  example.
\newblock {\em SERIEs\/}, {\em 2\/}(4), 497--513.

\bibitem[{Sui \& Boutilier(2015)}]{sui2015approximately}
Sui, X., \& Boutilier, C. (2015).
\newblock Approximately strategy-proof mechanisms for (constrained) facility
  location.
\newblock In {\em International Conference on Autonomous Agents and Multiagent
  Systems (AAMAS)\/}.

\bibitem[{Thomson(1985)}]{thomson1985trolley}
Thomson, J.~J. (1985).
\newblock The trolley problem.
\newblock {\em The Yale Law Journal\/}, {\em 94\/}(6), 1395--1415.

\bibitem[{Wang et~al.(2019)Wang, Zhao, Yu, Liu, Yang, Ren, \&
  Shi}]{wang2019privacy}
Wang, T., Zhao, J., Yu, H., Liu, J., Yang, X., Ren, X., \& Shi, S. (2019).
\newblock Privacy-preserving crowd-guided ai decision-making in ethical
  dilemmas.
\newblock In {\em ACM International Conference on Information and Knowledge
  Management (CIKM)\/}.

\bibitem[{Zeckhauser(1973)}]{zeckhauser1973voting}
Zeckhauser, R. (1973).
\newblock Voting systems, honest preferences and pareto optimality.
\newblock {\em American Political Science Review\/}, {\em 67\/}(3), 934--946.

\bibitem[{Zhang et~al.(2019)Zhang, Cheng, \& Conitzer}]{zhang2019better}
Zhang, H., Cheng, Y., \& Conitzer, V. (2019).
\newblock A better algorithm for societal tradeoffs.
\newblock In {\em Association for the Advancement of Artificial Intelligence
  (AAAI)\/}.

\end{thebibliography}

\clearpage
\appendix
\section{Concrete Example of Preference Vectors and Alternatives}
\label{appendix:vector-example}
\begin{figure}[H]
    \centering
    \includegraphics[width=0.7\textwidth]{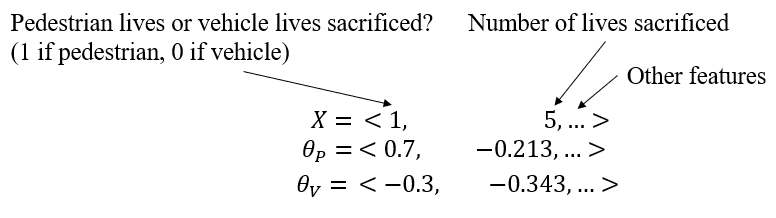}
    \caption{Examples of vectors used in a self-driving car ethics preference elicitation mechanism. $X$ is a vector representing an alternative with features corresponding to sacrificing five pedestrian lives (among other properties not displayed here). $\vec{\theta}_P$ is a preference vector of an individual who prioritizes sacrificing pedestrian lives over vehicle passenger lives. As a result, the vector has positive weight corresponding to that feature in the alternative representation. In contrast, $\vec{\theta}_V$ is a preference vector of an individual who prioritizes sacrificing vehicle passenger lives because the vector has negative weight corresponding to that feature. However, both individuals prefer sacrificing fewer lives because both have negative weights corresponding to the feature representing number of lives sacrificed.}
    \label{fig:mm-vectors}
\end{figure}

\section{Verbose Formulation of Equation \ref{eq:conditional-minority-agreement}}
\label{appendix:verbose-eq}

Equation \ref{eq:conditional-minority-agreement},
\begin{equation*}
    \mathbb{P}[\textrm{sign}({\vec{\theta}_C} \cdot (\vecl{y} - \vecl{x})) = \textrm{sign}({\vec{\theta}^*_\D} \cdot (\vecl{y} - \vecl{x}))\: |\:  
    \textrm{sign}({\vec{\theta}^*_\A} \cdot (\vecl{y} - \vecl{x})) \neq \textrm{sign}({\vec{\theta}^*_\D} \cdot (\vecl{y} - \vecl{x}))],
\end{equation*}
or the probability that the aggregate vector agrees with the {\dis} group's true preference vector 
for a random ordering \textit{given} that the true preference vectors
do not agree for that ordering, 
can also be written as
\begin{equation*}
        \frac{\mathbb{P}[\textrm{sign}({\vec{\theta}_C} \cdot (\vecl{y} - \vecl{x})) = \textrm{sign}({\vec{\theta}^*_\D} \cdot (\vecl{y} - \vecl{x})),
        \textrm{sign}({\vec{\theta}_C} \cdot (\vecl{y} - \vecl{x})) \neq \textrm{sign}({\vec{\theta}^*_\A} \cdot (\vecl{y} - \vecl{x}))]}{\mathbb{P}[\textrm{sign}({\vec{\theta}^*_\A} \cdot (\vecl{y} - \vecl{x})) \neq \textrm{sign}({\vec{\theta}^*_\D} \cdot (\vecl{y} - \vecl{x}))]}.
\end{equation*}
Note that this can be simplified further to
\begin{equation*}
    \frac{\mathbb{P}[
        \textrm{sign}({\vec{\theta}_C} \cdot (\vecl{y} - \vecl{x})) \neq \textrm{sign}({\vec{\theta}^*_\A} \cdot (\vecl{y} - \vecl{x}))]}{\mathbb{P}[\textrm{sign}({\vec{\theta}^*_\A} \cdot (\vecl{y} - \vecl{x})) \neq \textrm{sign}({\vec{\theta}^*_\D} \cdot (\vecl{y} - \vecl{x}))]},
\end{equation*}
because if the aggregate does not agree with the majority group's preference vector, it \textit{must} agree with the minority group's preference vector, but the reverse is not necessarily true. Therefore,
disagreement with the majority group's preference vector is a subset of agreement with the minority group's preference vector.
One can determine all of the quantities above via the equation in Proposition \ref{thm:agreement-prob} and the definition of $\vec{\theta}_C$ with $\vec{\theta}_{\A} = \vec{\theta}^*_{\A}$ and $\vec{\theta}_{\D} = \vec{\theta}^*_{\D}$. This yields
\begin{align*}
    &\frac{1 - \frac{\pi - \cos^{-1}(\vec{\theta}_C \cdot \vec{\theta}^*_\A)}{\pi}}{1 - \frac{\pi - \cos^{-1}(\vec{\theta}^*_\D \cdot \vec{\theta}^*_\A)}{\pi}},\\
    &=\frac{\frac{\cos^{-1}(\vec{\theta}_C \cdot \vec{\theta}^*_\A)}{\pi}}{\frac{ \cos^{-1}(\vec{\theta}^*_\D \cdot \vec{\theta}^*_\A)}{\pi}},\\
    &=\frac{\cos^{-1}(\vec{\theta}_C \cdot \vec{\theta}^*_\A)}{\cos^{-1}(\vec{\theta}^*_\D \cdot \vec{\theta}^*_\A)}.
\end{align*}

\section{Pure Strategy Nash Equilibrium Contained in 2D Plane: Formal Statement and Proof}
\label{appendix:two-dim}

\begin{proposition}
\label{proposition:two-dim}
Suppose $\vec{\theta}^*_i \in \mathbb{R}^d$ for each player $i \in \{\A,\D\}$, where $d > 2$, and $||\vec{\theta}^*_i||_2 = 1$. 
If a pure strategy Nash equilibrium
$(\vec{\theta}'_{\A},\vec{\theta}'_{\D})$ exists for this game, then the equilibrium $\vec{\theta}'_{\A},\vec{\theta}'_{\D}$ must be contained in $\Pi$, where $\Pi$ is the $2$-dimensional subspace of $\mathbb{R}^d$ spanned by $\vec{\theta}^*_{\A}, \vec{\theta}^*_{\D}$.
\end{proposition}

\begin{proof}

By Lemma \ref{cor:maj-consensus}, any pure strategy Nash equilibrium will have the property that
\begin{equation*}
    \vec{\theta}_C = \vec{\theta}^*_{\A} = \frac{(1-\alpha)\vec{\theta}'_{\A} + \alpha\vec{\theta}'_{\D}}{\lVert (1-\alpha)\vec{\theta}'_{\A} + \alpha\vec{\theta}'_{\D} \rVert}.    
\end{equation*}
The rest of this proof will be a proof by cases. Two cases are possible: 
\begin{enumerate}
    \item $\vec{\theta}'_{\A} \in \Pi$;
    \item $\vec{\theta}'_{\A} \not\in \Pi$.
\end{enumerate}

First, consider the case in which $\vec{\theta}'_{\A} \in \Pi$. In order to obtain $\vec{\theta}_C = \vec{\theta}^*_\A \in \Pi$ given that $\vec{\theta}'_{\A} \in \Pi$,
$\vec{\theta}'_{\D}$ must be contained in $\Pi$ as well. Therefore,
if $\vec{\theta}'_{\A} \in \Pi$, then a pure strategy Nash equilibrium $(\vec{\theta}'_{\A},\vec{\theta}'_{\D})$ is also contained in $\Pi$.

Second, consider the case in which $\vec{\theta}'_{\A} \notin \Pi$.
We show that if $\vec{\theta}'_{\A} \notin \Pi, \exists \vec{\theta}_{\D}: u_\D(\vec{\theta}'_{\A}, \vec{\theta}_{\D}) > u_\D(\vec{\theta}'_{\A}, \vec{\theta}'_{\D})$, which in turn means that $(\vec{\theta}'_{\A},\vec{\theta}'_{\D})$ is \emph{not} an equilibrium.

Suppose, via contradiction, $(\vec{\theta}'_{\A},\vec{\theta}'_{\D})$ is an equilibrium even though $\vec{\theta}'_{\A} \notin \Pi$. Even in this case,
because at equilibrium $\vec{\theta}_C = \vec{\theta}^*_{\A}$,
it must be possible for player $D$ to report $\vec{\theta}'_{\D}$ such that adding it to $\vec{\theta}'_{\A}$ results in angular movement from $\vec{\theta}'_{\A}$ towards $\vec{\theta}^*_{\A}$ equal to $\measuredangle \left(\vec{\theta}'_{\A}, \vec{\theta}^*_{\A} \right)$ (in the 2-dimensional subspace spanned by $\vec{\theta}'_{\A}$ and $\vec{\theta}^*_{\A}$).
However, player $D$ can instead report $\vec{\theta}_{\D}$ to cause angular movement of at least $\measuredangle \left(\vec{\theta}'_{\A}, \vec{\theta}^*_{\A} \right)$ \emph{directly towards $\vec{\theta}^*_{\D}$} (in the 2-dimensional subspace spanned by $\vec{\theta}'_{\A}$ and $\vec{\theta}^*_{\D}$).

If $\measuredangle \left(\vec{\theta}'_{\A}, \vec{\theta}^*_{\A} \right) \geq \measuredangle \left(\vec{\theta}'_{\A}, \vec{\theta}^*_{\D} \right)$, then $\exists \vec{\theta}_{\D}: \vec{\theta}_C = \frac{(1-\alpha)\vec{\theta}'_{\A} + \alpha\vec{\theta}_{\D}}{\lVert (1-\alpha)\vec{\theta}'_{\A} + \alpha\vec{\theta}_{\D} \rVert} = \vec{\theta}^*_{\D}$. This in turn would maximize utility for player $\D$, meaning that $u_\D(\vec{\theta}'_{\A}, \vec{\theta}_{\D}) > u_\D(\vec{\theta}'_{\A}, \vec{\theta}'_{\D})$. 

If instead $\measuredangle \left(\vec{\theta}'_{\A}, \vec{\theta}^*_{\A} \right) < \measuredangle \left(\vec{\theta}'_{\A}, \vec{\theta}^*_{\D} \right)$, then $\exists \vec{\theta}_{\D}: \vec{\theta}_C = \frac{(1-\alpha)\vec{\theta}'_{\A} + \alpha\vec{\theta}_{\D}}{\lVert (1-\alpha)\vec{\theta}'_{\A} + \alpha\vec{\theta}_{\D} \rVert}; \measuredangle \left(\vec{\theta}_C, \vec{\theta}^*_{\D}\right) = \measuredangle \left(\vec{\theta}'_{\A}, \vec{\theta}^*_{\D} \right) - \measuredangle \left(\vec{\theta}'_{\A}, \vec{\theta}^*_{\A} \right)$.

The shortest distance between two points on a surface of a sphere is the length of an arc of the great circle whose edge contains the two points, meaning a spherical version of the triangle inequality supports the following:
\begin{align*}
    \measuredangle \left(\vec{\theta}'_{\A}, \vec{\theta}^*_{\A} \right) + \measuredangle \left(\vec{\theta}^*_{\A}, \vec{\theta}^*_{\D} \right) &> \measuredangle \left(\vec{\theta}'_{\A}, \vec{\theta}^*_{\D} \right),
    \\
    \measuredangle \left(\vec{\theta}^*_{\A}, \vec{\theta}^*_{\D} \right) &> \measuredangle \left(\vec{\theta}'_{\A}, \vec{\theta}^*_{\D} \right) -  \measuredangle \left(\vec{\theta}'_{\A}, \vec{\theta}^*_{\A} \right).
\end{align*}
The left-hand side of the inequality above is inversely proportional to $u_{\D}(\vec{\theta}'_{\A}, \vec{\theta}'_{\D})$, and the right-hand side of the inequality is inversely proportional to $u_{\D}(\vec{\theta}'_{\A}, \vec{\theta}_{\D})$.Therefore, $u_\D(\vec{\theta}'_{\A}, \vec{\theta}_{\D}) > u_\D(\vec{\theta}'_{\A}, \vec{\theta}'_{\D})$ under these conditions as well.

The existence of a strategy $\vec{\theta}_\D \neq \vec{\theta}'_\D$ such that $u_\D(\vec{\theta}'_{\A}, \vec{\theta}_{\D}) > u_\D(\vec{\theta}'_{\A}, \vec{\theta}'_{\D})$ indicates a contradiction in the claim that there exists a pure strategy Nash equilibrium whenever $\vec{\theta}'_\A \notin \Pi$. Therefore, if a pure strategy Nash equilibrium $(\vec{\theta}'_\A,\vec{\theta}'_\D)$ exists, it must reside in $\Pi$.
\end{proof}

\section{Full Proof of Lemma \ref{thm:any-consensus}}
\label{appendix:maj-agg-extended-pf}
\begin{proof}
If we have $\vec{\theta}_C = \vec{\theta}^*_\A$, then it must be the case that $\frac{\alpha \vec{\theta}_D + (1-\alpha) \vec{\theta}_A}{\lVert \alpha \vec{\theta}_D + (1-\alpha) \vec{\theta}_A \rVert} = \vec{\theta}^*_\A$. This expression is equivalent to:
\begin{equation}\label{eq:theta_A}
    (1 - \alpha) \vec{\theta}_\A = \lVert \alpha \vec{\theta}_\D + (1 - \alpha) \vec{\theta}_\A \rVert \vec{\theta}^*_\A - \alpha \vec{\theta}_\D.
\end{equation}
We denote the magnitude of $\vec{\theta}^*_\A$ by $L$, that is, $L= \lVert \alpha \vec{\theta}_\D + (1 - \alpha) \vec{\theta}_\A \rVert$. (Note that $L\geq 0$ because it is the norm of a vector.)
So equation ~\ref{eq:theta_A} can be written as:
\begin{equation}\label{eq:theta_A_2}
    (1 - \alpha) \vec{\theta}_\A = L \vec{\theta}^*_\A - \alpha \vec{\theta}_\D.
\end{equation}
Figure \ref{fig:any-consensus-f} shows an illustration of equation~\ref{eq:theta_A_2}.
\begin{figure}[h]
    \centering
    \includegraphics[width=0.6\textwidth]{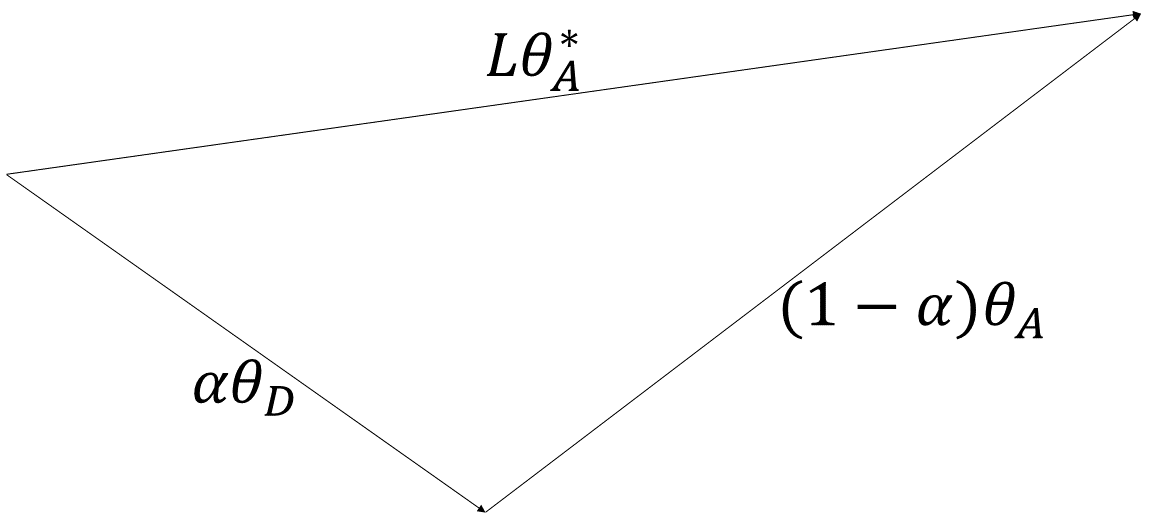}
    \caption{The illustration of $(1-\alpha) \vec{\theta}_\A = L\vec{\theta}^*_\A - \alpha \vec{\theta}_\D$. One can apply the Law of Cosines to solve for $L$.}
    \label{fig:any-consensus-f}
\end{figure}

Given that $\lVert \vec{\theta}_\A \rVert = \lVert \vec{\theta}^*_\A \rVert = \lVert \vec{\theta}_\D \rVert = 1$, after applying the Law of Cosines to equation~\ref{eq:theta_A_2}, we obtain:
\begin{equation}\label{eq:quadratic_L}
    (1 - \alpha)^2 = L^2 + \alpha^2  - 2\alpha L(\vec{\theta}_\D \cdot \vec{\theta}^*_\A),
\end{equation}
where $(\vec{\theta}_\D \cdot \vec{\theta}^*_\A)$ has been substituted for $\cos{\measuredangle \left(\vec{\theta}_\D, \vec{\theta}^*_\A \right)}$. Solving the above quadratic equation for $L$, we obtain: 
\begin{equation*}
    L = \alpha (\vec{\theta}_\D \cdot \vec{\theta}^*_\A) \pm \sqrt{\alpha^2 (\vec{\theta}_\D \cdot \vec{\theta}^*_\A)^2 - 2\alpha + 1}.
\end{equation*}
Note that $1-2\alpha > 0$, $\alpha (\vec{\theta}_\D \cdot \vec{\theta}^*_\A) < \sqrt{\alpha^2 (\vec{\theta}_\D \cdot \vec{\theta}^*_\A)^2 - 2\alpha + 1}$. Therefore $\alpha (\vec{\theta}_\D \cdot \vec{\theta}^*_\A) - \sqrt{\alpha^2 (\vec{\theta}_\D \cdot \vec{\theta}^*_\A)^2 - 2\alpha + 1} < 0$. This would be a contradiction with the fact that $L$ is positive. So it must be the case that:
\begin{equation*}
    L = \alpha (\vec{\theta}_\D \cdot \vec{\theta}^*_\A) + \sqrt{\alpha^2 (\vec{\theta}_\D \cdot \vec{\theta}^*_\A)^2 - 2\alpha + 1}.
\end{equation*}
Lastly, substituting $L$ in Equation~\ref{eq:theta_A_2} with the right hand side of the above equation 
results in
\begin{equation*}
    \vec{\theta}_\A = \frac{\bigg[\alpha (\vec{\theta}_\D \cdot \vec{\theta}^*_\A) + \sqrt{\alpha^2 (\vec{\theta}_\D \cdot \vec{\theta}^*_\A)^2 - 2\alpha + 1}\bigg]\vec{\theta}^*_\A - \alpha \vec{\theta}_\D}{(1 - \alpha)}.
\end{equation*}
\end{proof}

\section{Full Proof of Lemma \ref{lem:minority-ort}}
\label{appendix:dist-deriv}
\begin{proof}
As in Appendix \ref{appendix:maj-agg-extended-pf}, consider the vector sum $\alpha \vec{\theta}_\D + (1-\alpha) \vec{\theta}_\A = L\vec{\theta}_C$, where $L = \lVert \alpha \vec{\theta}_\D + (1 - \alpha) \vec{\theta}_\A \rVert$. We consider the general case where $\vec{\theta}_C$ does not necessarily equal $\vec{\theta}^*_\A$. Figure \ref{fig:minority-ort-appendix} shows an illustration of this sum.

\begin{figure}[h]
    \centering
    \includegraphics[width=0.6\textwidth]{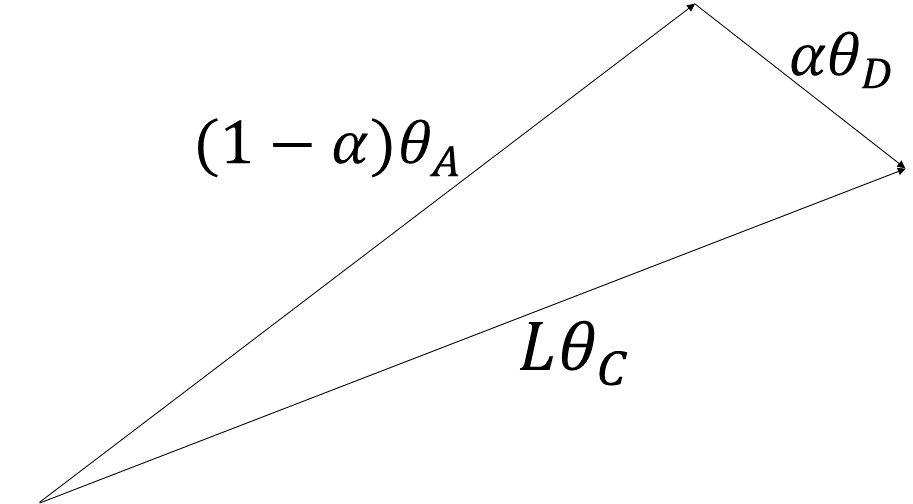}
    \caption{Illustration of $\alpha \vec{\theta}_\D + (1-\alpha) \vec{\theta}_\A = L\vec{\theta}_C$. One can apply the Law of Sines to solve for $\measuredangle \left(\vec{\theta}_C, \vec{\theta}_{\A}\right)$.}
    \label{fig:minority-ort-appendix}
\end{figure}

According to the Law of Sines, the following must be true:
\begin{equation*}
    \frac{\alpha}{\sin{(\measuredangle \left(\vec{\theta}_C, \vec{\theta}_{\A}\right)})} = \frac{1-\alpha}{\sin{(\measuredangle \left(\vec{\theta}_C, \vec{\theta}_{\D}\right)})}.
\end{equation*}
Therefore,
\begin{equation*}
    \measuredangle \left(\vec{\theta}_C, \vec{\theta}_{\A}\right) = \sin^{-1}{\left(\frac{\alpha \sin{(\measuredangle \left(\vec{\theta}_C, \vec{\theta}_{\D}\right)})}{1-\alpha}\right)}.
\end{equation*}

Now, three cases are possible:
(i) $0 \leq \measuredangle \left(\vec{\theta}_C, \vec{\theta}_\D \right) < \frac{\pi}{2}$, 
(ii) $\measuredangle \left(\vec{\theta}_C, \vec{\theta}_\D \right) = \frac{\pi}{2}$, and
(iii) $\frac{\pi}{2} < \measuredangle \left(\vec{\theta}_C, \vec{\theta}_\D \right) \leq \pi$. A visualization of each of these cases is shown in Figure \ref{fig:minority-ort}.

\begin{figure}[H]
    \centering
    \includegraphics[width=0.6\textwidth]{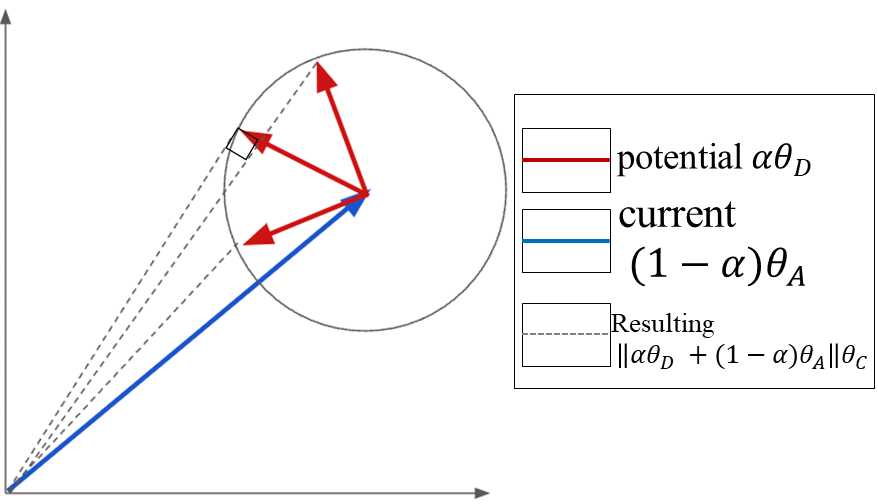}
    \caption{Illustrations of $\alpha \vec{\theta}_\D + (1-\alpha)\vec{\theta}_\A = \lVert \alpha \vec{\theta}_D + (1-\alpha) \vec{\theta}_A \rVert \theta_C$ for various $\alpha \vec{\theta}_\D$ candidates (red) and results (dashed) relative to $(1-\alpha)\vec{\theta}_\A$ (blue). Maximum change in direction from $\vec{\theta}_A$ to $\vec{\theta}_C$ when $\vec{\theta}_\D \perp \vec{\theta}_C$.}
    \label{fig:minority-ort}
\end{figure}
 
 Note that in cases (i) and (iii), $\measuredangle \left(\vec{\theta}_C, \vec{\theta}_\A \right) = \sin^{-1}{\left(\frac{\alpha \sin{(\measuredangle \left(\vec{\theta}_C, \vec{\theta}_\D \right))}}{1-\alpha}\right)} < \sin^{-1}{\left( \frac{\alpha}{1-\alpha}\right)}$ because $\left(\frac{\alpha \sin{(\measuredangle \left(\vec{\theta}_C, \vec{\theta}_\D \right))}}{1-\alpha}\right) < \left(\frac{\alpha}{1-\alpha}\right)$ and because $\sin^{-1}(.)$ is a monotonic function in $[-1,1]$.
 
 In case (ii),  $\measuredangle \left(\vec{\theta}_C, \vec{\theta}_\A \right) = \sin^{-1}{\left(\frac{\alpha \sin{(\measuredangle \left(\vec{\theta}_C, \vec{\theta}_\D \right))}}{1-\alpha}\right)} = \sin^{-1}{\left(\frac{\alpha}{1-\alpha}\right)}$. Therefore, in all cases, $\measuredangle \left(\vec{\theta}_C, \vec{\theta}_\A \right) \leq \sin^{-1}{\left( \frac{\alpha}{1-\alpha}\right)}$. Equality is obtained if and only if $\vec{\theta}_\D \perp \vec{\theta}_C$.

\end{proof}

\section{Additional Lemma Proofs}
\label{appendix:corr-proofs}

\subsection{Proof of Lemma \ref{cor:maj-consensus}}
\begin{proof}
Suppose, by contradiction, there exists a pure strategy Nash equilibrium $(\vec{\theta}'_\D, \vec{\theta}'_\A)$ for $G$, and $\vec{\theta}_C \neq \vec{\theta}^*_{\A}$. By Lemma \ref{thm:any-consensus}, there must exist a $\vec{\theta}_\A$ given the {\dis} group's best response $\vec{\theta}'_\D$ such that $\vec{\theta}_C = \vec{\theta}_\A^*$. Note that such a $\vec{\theta}_{\A}$ maximizes the {\adv} group's utility. Therefore, $\vec{\theta}_\A$ produces a better outcome for the {\adv} group than $\vec{\theta}'_\A$. However, this raises a contradiction because if $(\vec{\theta}'_\D, \vec{\theta}'_\A)$ is a pure strategy Nash equilibrium, $\vec{\theta}'_\A$ must be a best response, but $\vec{\theta}_\A$ is a better response. Thus, for every pure strategy Nash equilibrium $(\vec{\theta}'_\D, \vec{\theta}'_\A)$, $\vec{\theta}_C = \vec{\theta}^*_{\A}$.

\end{proof}

\subsection{Proof of Lemma \ref{cor:min-ub}}
\begin{proof}
According to Lemma \ref{lem:minority-ort}, 
the {\dis} group can pull the aggregate away from $\vec{\theta}_\A$ by $\sin^{-1}{\left(\frac{\alpha}{1-\alpha}\right)}$. In this case, we have $\measuredangle \left(\vec{\theta}_C, \vec{\theta}^*_{\D}\right) = \left(\vec{\theta}_A, \vec{\theta}^*_{\D}\right) - \sin^{-1}\left(\frac{\alpha}{1-\alpha}\right)$.
By definition, the largest possible value for $\measuredangle \left(\vec{\theta}_{\A}, \vec{\theta}^*_{\D}\right)$ is $\pi$. Therefore, we have 
$$\measuredangle \left(\vec{\theta}_C, \vec{\theta}^*_{\D}\right) \leq \pi - \sin^{-1}\left(\frac{\alpha}{1-\alpha}\right).$$
\end{proof}

\subsection{Proof of Lemma \ref{lem:min-orth-eq}}
\begin{proof}
Suppose, by contradiction, there \textit{is} a pure strategy Nash equilibrium where the {\dis} group's best response $\vec{\theta}'_{\D}$ is \textit{not} orthogonal to $\vec{\theta}_C$.
Consider the {\adv} group's best response $\vec{\theta}'_{\A}$. By Lemma \ref{lem:minority-ort}, the most the {\dis} can move the aggregate from $\vec{\theta}'_{\A}$ is by $\sin^{-1}\left(\frac{\alpha}{1-\alpha}\right)$. Therefore, it is possible to obtain any aggregate vector in the cone centered at $\vec{\theta}'_{\A}$ with apex angle $2\sin^{-1}\left(\frac{\alpha}{1-\alpha}\right)$, including the vectors located on its borders. Note that according to Lemma \ref{lem:minority-ort}, it is only possible to obtain the aggregate vectors on the borders of this cone if $\vec{\theta}'_{\D} \perp \vec{\theta}_C$.  Depending on the location of $\vec{\theta}^*_{\D}$ relative to this cone, there are three cases:
\begin{enumerate}
    \item $\vec{\theta}^*_{\D}$ is strictly inside of this cone;
    \item $\vec{\theta}^*_{\D}$ is on a border of this cone;
    \item $\vec{\theta}^*_{\D}$ is outside of this cone.
\end{enumerate}
See Figure \ref{fig:lemma_5_v2} for a visualization of each of these cases.

In all cases that follow, $\vec{\theta}_C = \vec{\theta}^*_{\A}$ in this pure strategy Nash equilibrium according to Lemma~\ref{cor:maj-consensus}. Additionally, note that if $\vec{\theta}_C$ were on a border of this cone, $\vec{\theta}'_{\D}$ would be orthogonal to $\vec{\theta}_C$ according to Lemma \ref{lem:minority-ort}. However, this would be a contradiction with the given assumption. Therefore, $\vec{\theta}_C = \vec{\theta}^*_{\A}$ is strictly inside of this cone.

\begin{figure}
    \centering
    \includegraphics[width=0.95\columnwidth]{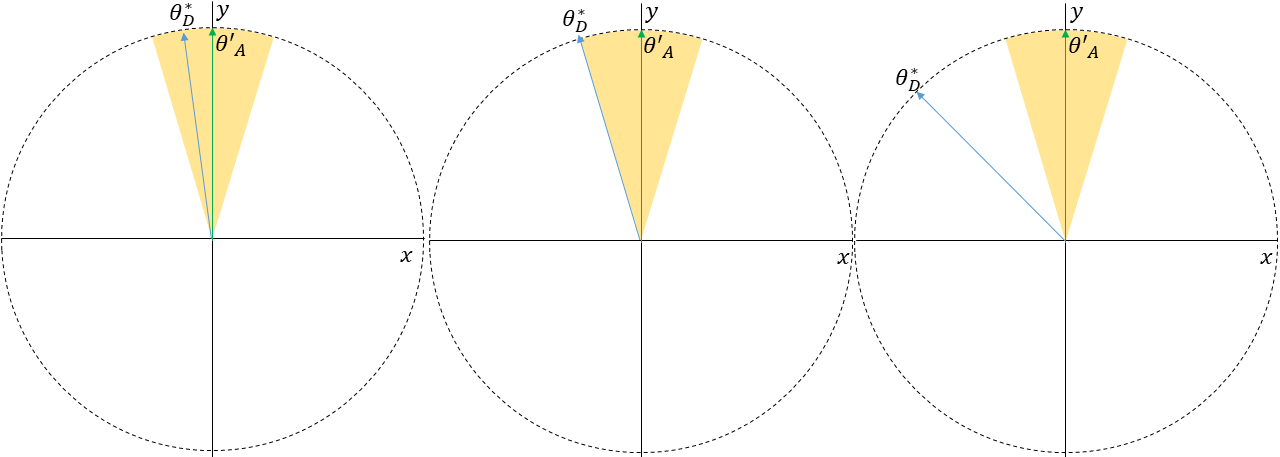}
    \caption{Illustrations of various $\vec{\theta}^*_{\D}$ relative to $\vec{\theta}'_{\A}$ and cone of possible aggregate vectors (yellow) with central angle $2\sin^{-1}\left(\frac{\alpha}{1-\alpha}\right)$.}
    \label{fig:lemma_5_v2}
\end{figure}

First, consider the case in which $\vec{\theta}^*_{\D}$ is strictly inside of this cone. This means that $\exists \vec{\theta}_{\D}: \vec{\hat{\theta}}_C = \frac{(1-\alpha)\vec{\theta}'_{\A} + \alpha\vec{\theta}_{\D}}{\lVert (1-\alpha)\vec{\theta}'_{\A} + \alpha\vec{\theta}_{\D} \rVert} = \vec{\theta}^*_{\D}$, where $\vec{\hat{\theta}}_C$ is the aggregate obtained from $\vec{\theta}_{\D}$ and $\vec{\theta}'_{\A}$. As mentioned in Section \ref{sec:setup} we also assume that $\vec{\theta}^*_{\A} \neq \vec{\theta}^*_{\D}$. Therefore, such a $\vec{\theta}_{\D}$ would be a better response than $\vec{\theta}'_{\D}$. However, this raises a contradiction because $\vec{\theta}'_{\D}$ is supposed to be a best response, but reporting $\vec{\theta}_\D$ such that $\vec{\hat{\theta}}_C = \vec{\theta}^*_{\D}$ would be a better response. This means that such a pure strategy Nash equilibrium cannot exist in this case.

Second, consider the case in which $\vec{\theta}^*_{\D}$ is on a border of this cone. In this case, $\exists \vec{\theta}_{\D}: \vec{\theta}_{\D} \perp \vec{\hat{\theta}}_C; \vec{\hat{\theta}}_C = \frac{(1-\alpha)\vec{\theta}'_{\A} + \alpha\vec{\theta}_{\D}}{\lVert (1-\alpha)\vec{\theta}'_{\A} + \alpha\vec{\theta}_{\D} \rVert} = \vec{\theta}^*_{\D}$, where $\vec{\hat{\theta}}_C$ is the aggregate obtained from $\vec{\theta}_{\D}$ and $\vec{\theta}'_{\A}$. This also raises a contradiction because there exists a better response than $\vec{\theta}'_{\D}$ (that obtains $\vec{\hat{\theta}}_C = \vec{\theta}^*_{\D}$) for the {\dis} group that \textit{is} orthogonal to the aggregate vector $\vec{\hat{\theta}}_C$ that results. In turn, such a pure strategy Nash equilibrium cannot exist in this case either.

Lastly, consider the case in which $\vec{\theta}^*_{\D}$ is outside of this cone. In this case, $\exists \vec{\theta}_{\D}: \vec{\theta}_{\D} \perp \vec{\hat{\theta}}_C; \vec{\hat{\theta}}_C = \frac{(1-\alpha)\vec{\theta}'_{\A} + \alpha\vec{\theta}_{\D}}{\lVert (1-\alpha)\vec{\theta}'_{\A} + \alpha\vec{\theta}_{\D} \rVert}$, where $\vec{\hat{\theta}}_C$ is the aggregate obtained from $\vec{\theta}_{\D}$ and $\vec{\theta}'_{\A}$, and $\vec{\hat{\theta}}_C$ is on a border of the cone. More specifically, there are two such $\vec{\hat{\theta}}_C$ because there are two borders. Recall regarding the original aggregate that $\vec{\theta}_C = \vec{\theta}^*_{\A}$. Moreover, $\vec{\theta}_C = \vec{\theta}^*_{\A}$ is strictly inside of this cone. As a result, one of the borders will be closer to $\vec{\theta}^*_{\D}$ than $\vec{\theta}^*_{\A}$. Therefore, at least one $\vec{\hat{\theta}}_C$ will satisfy $\vec{\hat{\theta}}_C \cdot \vec{\theta}^*_\D > \vec{\theta}^*_{\A} \cdot \vec{\theta}^*_{\D}$. This raises a contradiction as well because there exists a better response than $\vec{\theta}'_{\D}$ for the {\dis} group that \textit{is} orthogonal to the aggregate vector $\vec{\hat{\theta}}_C$ that results. It is therefore evident that such a pure strategy Nash equilibrium cannot exist in this case as well.

Thus, via proof by contradiction and cases, in any pure strategy Nash equilibrium, the {\dis} group's best response is orthogonal to $\vec{\theta}_C$. Moreover, in order to obtain $\vec{\theta}_C = \vec{\theta}^*_{\A}$ according to Lemma \ref{cor:maj-consensus} with $\vec{\theta}'_\D \perp \vec{\theta}_{C}$, $\vec{\theta}'_{\A}$ will offset their response from their true preferences $\vec{\theta}^*_{\A}$ by $\sin^{-1}\left(\frac{\alpha}{1-\alpha}\right)$.
\end{proof}

\section{Algebraic Simplification for Proof of Theorem \ref{thm:nondegen-eq}}
\label{appendix:eq-alg}

Equation \ref{eq:maj-group-any-direction} is given by
\begin{equation*}
    \vec{\theta}_\A = \frac{\bigg[\alpha (\vec{\theta}_\D \cdot \vec{\theta}^*_\A) + \sqrt{\alpha^2 (\vec{\theta}_\D \cdot \vec{\theta}^*_\A)^2 - 2\alpha + 1}\bigg]\vec{\theta}^*_\A - \alpha \vec{\theta}_\D}{(1 - \alpha)}.
\end{equation*}
Whenever $\vec{\theta}_\D \perp \vec{\theta}_C$,
\begin{align*}
    \alpha (\vec{\theta}_\D \cdot \vec{\theta}^*_\A) + \sqrt{\alpha^2 (\vec{\theta}_\D \cdot \vec{\theta}^*_\A)^2 - 2\alpha + 1} &= \lVert \alpha \vec{\theta}_\D + (1 - \alpha) \vec{\theta}_\A \rVert \\
    &= \sqrt{(1-\alpha)^2 - \alpha^2} \\
    &= \sqrt{1 - 2\alpha}.
\end{align*}
$\vec{\theta}_\D \perp \vec{\theta}_C$ in this case because $\vec{\theta}_\D = \vec{\theta}'_\D = \sign(\vec{\theta}^*_\D \cdot R \vec{\theta}^*_{\A})R \vec{\theta}^*_\A$ and $\vec{\theta}^*_\A = \vec{\theta}_C$. Substitution of values into Equation \ref{eq:maj-group-any-direction} yields
\begin{align*}
    \vec{\theta}'_\A &= \frac{\left(\sqrt{1 - 2\alpha}\right)\vec{\theta}^*_\A - \alpha \sign(\vec{\theta}^*_\D \cdot R \vec{\theta}^*_{\A})R \vec{\theta}^*_\A}{(1 - \alpha)}, \\
    &= \frac{\big[(\sqrt{1 - 2\alpha})I - \alpha \sign(\vec{\theta}^*_\D \cdot R \vec{\theta}^*_{\A})R\big]\vec{\theta}^*_\A}{1 - \alpha}.
\end{align*}

\end{document}